\documentclass[journal]{IEEEtran}

\usepackage[T1]{fontenc}
\usepackage[utf8]{inputenc}

\usepackage{amsmath,amssymb,amsfonts,amsthm,bm,mathtools}

\newtheorem{lemma}{Lemma}
\newtheorem{proposition}{Proposition}

\ifCLASSINFOpdf
\usepackage[pdftex]{graphicx}
\else
\usepackage{graphicx}
\fi

\usepackage{cite}
\usepackage{float}
\usepackage{stfloats}            
\usepackage{algorithm}
\usepackage{algpseudocode}       

\usepackage{booktabs}
\usepackage{multirow}
\usepackage{tabularx}
\usepackage{adjustbox}
\usepackage[table,xcdraw]{xcolor}

\usepackage[font=scriptsize]{caption}
\usepackage{subcaption}          

\usepackage{url}
\usepackage{hyperref}

\usepackage{array}
\usepackage{color,soul}          
\usepackage{pifont}              
\newcommand{\cmark}{\ding{51}}%
\newcommand{\xmark}{\ding{55}}%

\usepackage{varwidth}
\usepackage[compatibility=false]{caption}

\begin{document}

\title{STAR-RIS-Aided Secure Communications: Analytical Insights and Performance Comparison}
\author{Taissir Y. Elganimi,~\IEEEmembership{Senior Member,~IEEE},
Mahmoud Aldababsa,~\IEEEmembership{Senior Member,~IEEE},
Ali A. Nasir,~\IEEEmembership{Senior Member,~IEEE},
and Khaled M. Rabie,~\IEEEmembership{Senior Member,~IEEE}

\thanks{

T. Elganimi is with the Department of Electrical and Electronic Engineering, University of Tripoli, Libya (e-mail: t.elganimi@uot.edu.ly).

M. Aldababsa is with the Department of Electrical and Electronics Engineering, Nisantasi University, 34481742, Istanbul, Turkey (e-mail: mahmoud.aldababsa@nisantasi.edu.tr).

A. Nasir is with the Department of Electrical Engineering and the Center for Communication Systems and Sensing, King Fahd University of Petroleum and Minerals (KFUPM), Dhahran 31261, Saudi Arabia (e-mail: anasir@kfupm.edu.sa).

K. Rabie is with the Department of Computer Engineering and the Center for Communication Systems and Sensing, King Fahd University of Petroleum and Minerals (KFUPM), Dhahran 31261, Saudi Arabia (e-mail: k.rabie@kfupm.edu.sa).
}}

\markboth{Submitted to IEEE Internet of Things Journal,~Vol.~xx, No.~xx, XXXXX~2025}
{Shell \MakeLowercase{\textit{et al.}}: A Sample Article Using IEEEtran.cls for IEEE Journals}



\maketitle

\begin{abstract}
Simultaneously transmitting and reflecting reconfigurable intelligent surfaces (STAR-RISs) have emerged as a promising technology for enabling full-space signal manipulation and enhancing wireless network coverage and capacity. In this article, we present a comprehensive analytical comparison of STAR-RIS-assisted systems with single-input single-output (SISO), conventional RISs, and decode-and-forward (DF) relaying schemes, including both half-duplex (HD) and full-duplex (FD) modes. Closed-form expressions are derived for the achievable secrecy rates of STAR-RIS-aided communications under both the absence and presence of eavesdroppers. Unlike most existing works, the direct source–destination link is incorporated in all considered schemes, and optimal transmit power allocation is investigated for HD- and FD-DF relaying. Furthermore, we provide the conditions under which STAR-RIS outperforms HD- and FD-DF relaying and quantify the minimum number of STAR-RIS elements required to achieve superior rates. The impacts of key system parameters—including transmit power, number of elements, reflection-to-transmission power ratio, element-splitting factor, and deployment positions—on both achievable and secrecy performance are investigated. The results reveal that STAR-RIS systems can achieve superior rates and secrecy rates compared to all benchmark schemes.
\end{abstract}


\section{Introduction}

\IEEEPARstart{R}{ecently}, reconfigurable intelligent surfaces (RISs) have been deemed a key enabling technology for sixth-generation (6G) wireless communication systems, owing to their ability to significantly improve the achievable data rates and energy efficiency (EE) \cite{8910627}. RISs have the potential to enable the emerging paradigm of smart radio environments by exploiting the unique properties of metamaterials and large-scale arrays of low-cost passive elements. They assist wireless systems by transforming the wireless propagation environment into an intelligently reconfigurable and software-controllable space. This is achieved by carefully adjusting the phase shifts of numerous low-cost passive reflecting elements integrated within the RIS \cite{9475160}. In wireless communication networks, RIS can be employed to redirect incident signals towards the destination (D), thereby improving the reception quality and link performance \cite{REF1}.  Consequently, RIS technology  fulfills the requirement for a smart radio environment and possesses the ability to improve the effectiveness of data transmission systems \cite{8796365, 9140329}.

Different from conventional relaying protocols, RISs passively reflect incident signals without requiring expensive radio frequency (RF) chains, in contrast to decode-and-forward (DF) and amplify-and-forward (AF) relaying schemes that operate in half-duplex (HD) mode, and require active signal processing. Compared to conventional active transmitters, RISs offer substantially more cost-effective operation, characterized by reduced transceiver hardware requirements and lower power consumption. Moreover, RISs inherently operate in full-duplex (FD) mode without suffering from self-interference or introducing additional thermal noise due to their passive reflection mechanism. Furthermore, RISs offer superior spectral efficiency compared to active HD relays, while maintaining a lower signal processing complexity than active FD relays, which suffer from residual self-interference and require advanced cancellation techniques \cite{9475160}. Therefore, RISs offer a more cost-effective alternative to the classical relaying and multi-antenna technologies.

In general, conventional RISs are typically deployed on one side of objects such as walls, ceilings, building facades, street lamps, etc, and are limited to passively reflecting incident signals toward the same side as the source (S). This, however, leads to a half-space smart radio environment, requiring both the S and the users to be located on the same side of the RIS. In practical deployments, users are typically distributed on both sides of the RIS, which significantly restricts both the flexibility and overall effectiveness of conventional RIS configurations. Therefore, simultaneously transmitting and reflecting RISs (STAR-RISs) have been proposed in \cite{REF_STAR1_4} to overcome the inherent half-space coverage limitation of conventional RISs, thereby enabling intelligent signal manipulation across the entire coverage area and serving users located on both sides of the surface. The STAR-RIS is functionally divided into two parts; one part reflects the incident signal towards users in the same space, referred to as the reflection space, while the other part transmits the signal to the opposite space, known as the transmission space. In STAR-RISs, the transmission and reflection coefficients are used to simultaneously reconfigure both the transmitted and reflected signals, achieving full-space coverage and enabling smart control of the wireless environment across all spatial directions \cite{REF_STAR1_4, 9690478}.

Despite the promising advantages of RISs and STAR-RISs, and considering that RISs, STAR-RISs, and various types of relays all provide alternative transmission links, it is essential to compare these emerging technologies with conventional relaying schemes. Previous comparative studies have primarily focused on scenarios where RISs reflect incident signals \cite{REF222, 8811733, 9095301, REF_Emil_3, 9184098, 9217321, 9154308, 9119122, 9703626, 9625201, 9359653, 9569598, 10015640, 10.1007/978-3-031-28076-4_10, 10211119, 10183245, 10373918, 10598369, JIANG2024514, 10963359, 10811893, 26262626}, with only one study \cite{11142285} addressing STAR-RIS, as highlighted in the following subsection.

\subsection{Related Work}
In \cite{REF222}, RIS and HD-AF relaying are compared, and the use of RIS is shown to produce significant improvements in the EE performance compared to conventional AF relaying schemes. In another work \cite{8811733}, it was shown that the RIS-aided system achieves a higher rate and greater cost-efficiency compared to HD- and FD-AF relaying systems when the number of reflecting elements is sufficiently large. In \cite{9095301}, a complete performance analysis of RIS-assisted systems is provided including the instantaneous and average end-to-end signal-to-noise ratio (SNR), the outage probability (OP), the symbol error rate (SER), the diversity gain, the diversity order, and the ergodic capacity. It is shown in \cite{9095301} that RIS-assisted systems outperform HD-AF relaying systems in terms of average SNR, OP, SER, and ergodic capacity.

The authors in \cite{REF_Emil_3} compared RIS with the classical HD-DF relaying scheme in terms of both the achievable rate and EE performance using the radio propagation model that has been defined by the third generation partnership project (3GPP). The results reported in \cite{REF_Emil_3} showed that RIS does not consistently outperform HD-DF relaying; however, the former can achieve performance superior to that of the latter when a large number of reflecting elements are employed. In another comparative study \cite{9184098}, the authors compared RIS with the HD-DF relaying scheme in terms of information rate, where the results also showed that RIS can achieve performance superior to that of HD-DF relays in the case of adequately large RISs. Different from these two comparative studies, the author in \cite{9217321} compared the performance of RIS, single-input single-output (SISO), and HD-DF relays, and derived closed-form expressions of the achievable rates and transmit powers to determine the required number of reflecting elements in an RIS to outperform HD-DF relaying. The main finding in \cite{9217321} is that both RIS and HD-DF relays can complement each other’s strengths and have the potential for integration within the fifth-generation (5G) and beyond-5G (B5G) network architectures. In \cite{9154308}, the authors compared the achievable rate and EE performance of RIS and HD-DF relaying, taking into account the insertion losses and power consumption associated with the electronic components of the deployed nodes. The results reported in \cite{9154308} also showed that the RIS-aided system outperforms its relay-aided counterparts when the RIS is sufficiently large in size. In \cite{9119122}, the authors reviewed the key differences and similarities between HD- and FD-DF relaying and RISs, and showed that large RISs are capable of outperforming relay-based systems in terms of the achievable data rate. Later, the authors in \cite{9703626} investigated the performance of HD-DF relaying and RIS in unmanned aerial vehicle (UAV)-assisted communication scenarios. This study analyzed different system configurations and requirements in order to compare the effectiveness of RIS and relaying systems under different operational conditions. In another study \cite{9625201}, the performance of RIS is compared to that of HD- and FD-DF relaying for general multiple-input multiple-output (MIMO) systems by optimizing the beamforming matrices. This study highlighted the ability of RIS to offer comparable spectral efficiency to HD relaying while providing significantly higher EE than FD relaying. In \cite{9359653}, both HD- and FD-DF relaying systems are compared with spatially-distributed RISs in terms of OP and EE performance, where only one relay (R) or RIS is selected based on maximizing the SNR. It is shown in \cite{9359653} that RIS-based systems offer a more energy-efficient solution for assisting communication than relaying schemes. Furthermore, RIS can enhance both OP and EE performance, especially when equipped with a large number of reflecting elements. In a further comparison-based study \cite{9569598}, the authors presented a performance comparison between an RIS-assisted system and FD-DF relay-assisted system in the presence of a non-ideal transmitter. The main finding in \cite{9569598} is that the RIS-assisted system can never achieve better capacity performance than the traditional FD-DF relaying, regardless of the position of the RIS or the relay. It is also shown that the capacity of the RIS-assisted system approaches the channel capacity provided by the FD-DF relay-assisted system as the number of reflecting elements increases. However, this results in a reduction in the EE performance of the RIS-assisted system, which involves making a trade-off between channel capacity enhancement and EE degradation in the RIS-assisted system.

Recently, the authors in \cite{10015640} compared the performance of single RF chain multi-antenna FD-DF relaying and passive RIS, where the main findings showed that the former demonstrates significantly higher EE compared to the latter. The authors in \cite{10.1007/978-3-031-28076-4_10} compared RIS and relaying systems particularly focusing on UAV communications scenarios. The results reported in \cite{10.1007/978-3-031-28076-4_10} highlight that the active relays provide higher EE than the RIS when the UAV is located near the midpoint between the base station and the user, while the RIS outperforms active relays when the UAV is positioned close to either the base station or the user. In addition, it is shown that DF relaying offers higher EE than AF relaying. In \cite{10211119} and \cite{10183245}, the performance of RIS, HD and FD relays is compared from the operator perspective, and closed-form expressions of the achievable rates for RIS and active relays are derived to make a fair comparison. The results in \cite{10211119} and \cite{10183245} highlight that the RIS and FD-DF relaying are effective in enhancing the system performance. The superiority in the achievable rate between the two technologies depends on the number of reflecting elements in the RIS and the level of self-interference suppression in the FD-DF relay. Later, RIS is analytically compared with FD-DF and FD-AF relays in UAV cooperative communications \cite{10373918}. The numerical results in \cite{10373918} indicate that RIS of moderate size can achieve performance comparable to that of AF relays, while surpassing DF relays under conditions of high data rates and deploying a large number of reflecting elements. In another work for the same authors \cite{10598369}, the performance of RISs and FD relays is compared in a millimetre-wave MIMO network. It is found in \cite{10598369} that RISs offer higher EE at sub-6 GHz, whereas FD relays provide better spectral efficiency than that of RISs. Besides, DF relays usually outperform their AF counterparts. In \cite{JIANG2024514}, a secure performance comparison between RIS and HD-AF relaying is presented for non-orthogonal multiple access (NOMA) with the presence of an eavesdropper (E). In this work, the similar number of RIS elements and AF relay antennas is assumed for the sake of fair comparison, and the same secure transport strategies are used for both system models in order to maximize the secrecy rate. The simulation results in \cite{JIANG2024514} show that the RIS-assisted NOMA outperforms AF relay-assisted NOMA when the number of elements reaches a certain value.

More recently, the authors in \cite{10963359} investigated the impact of RIS deployment on improving both power efficiency and sustainability in wireless communications and proposed optimization frameworks to minimize the power consumption in RIS-assisted systems. In addition, a comparative analysis between RIS and classic active relaying systems is provided in \cite{10963359} and showed that RIS provides performance gains over HD-AF and HD-DF relaying. In a recent work \cite{10811893}, novel multiple access communication protocols are proposed for cooperative transmission in industrial Internet of Things subnetworks, leveraging secondary access points as HD-DF and HD-AF relays. These protocols are extended to include RISs. The results showed that the RIS-based protocol attains superior power savings, while the relay-based solution provides better OP than that of RIS-based solution. In another recent comparative study \cite{26262626}, the authors presented a comparative analysis between RIS and both DF and AF relaying schemes. The analysis in \cite{26262626} demonstrates that RIS-aided systems with optimized phase shifts are capable of achieving substantial power savings and superior EE compared to HD-AF and HD-DF relaying. Furthermore, RIS-aided systems require fewer reflecting elements to outperform relaying schemes under high achievable rate demands. This makes it an attractive choice for high-performance wireless system design. In a more recent work \cite{11142285}, the STAR-RIS system is compared against both HD-DF and FD-DF relaying. In this work, new analytical conditions are derived to compute the number of STAR-RIS elements required to outperform HD- and FD-DF relaying schemes, and provide new insights into the interplay between residual loop interference and transmission/reflection splits. Therefore, while the overall trend that large surfaces are needed is consistent with \cite{REF_Emil_3}, the comparison is extended in \cite{11142285} to include STAR-RIS and FD-DF relaying systems, given the STAR-RIS's ability to simultaneously reflect and transmit signals, resembling the FD operation and offering deeper insights into their performance trade-offs.

A comparative summary of the key characteristics of prior research and the present work is presented in Table \ref{tbl1}, focusing on the RIS scheme, relaying scheme, duplex mode, and secrecy performance.

\begin{table*} [t]
\caption{A comparison between the related and present work.}
\label{tbl1}
\begin{adjustbox}{width=\textwidth}
\begin{tabular}{ |l|c|c|c|c|c| } 
\noalign{\hrule height 1pt}

\textbf{Reference} & Year & RIS scheme & Relaying scheme & Duplex mode  & Secrecy performance  \\ \hline
 
\cite{REF222} & 2019 & RIS & AF & HD & \xmark  \\  \hline
 
\cite{8811733} & 2019 & RIS & AF & HD and FD & \xmark \\  \hline
 
\cite{9095301} & 2020 & RIS & AF & HD & \xmark \\  \hline
  
\cite{REF_Emil_3} & 2020 & RIS & DF & HD & \xmark \\  \hline
 
\cite{9184098} & 2020 & RIS & DF & HD & \xmark \\  \hline
 
\cite{9217321} & 2020 & RIS & DF & HD & \xmark  \\  \hline
 
\cite{9154308} & 2020 & RIS & DF & HD & \xmark \\  \hline

\cite{9119122} & 2020 & RIS & DF & HD and FD & \xmark  \\  \hline
 
\cite{9703626} & 2021 & RIS & DF & HD & \xmark \\  \hline
 
\cite{9625201} & 2021 & RIS & DF & HD and FD & \xmark  \\  \hline
 
\cite{9359653} & 2021 & RIS & DF & HD and FD & \xmark  \\  \hline

\cite{9569598} & 2021 & RIS & DF & HD and FD & \xmark  \\  \hline

\cite{10015640} & 2023 & RIS & DF & HD and FD & \xmark  \\  \hline

\cite{10.1007/978-3-031-28076-4_10} & 2023 & RIS & AF and DF & FD & \xmark  \\  \hline

\cite{10211119} & 2023 & RIS & DF & HD and FD & \xmark  \\  \hline

\cite{10183245} & 2023 & RIS & DF & HD and FD & \xmark  \\  \hline

\cite{10373918} & 2024 & RIS & AF and DF & FD & \xmark  \\  \hline

\cite{10598369} & 2024 & RIS & AF and DF & FD & \xmark  \\  \hline

\cite{JIANG2024514} & 2024 & RIS & AF & HD & \cmark  \\  \hline

\cite{10963359} & 2025 & RIS & AF and DF & HD & \xmark  \\  \hline

\cite{10811893} & 2025 & RIS & AF and DF & HD & \xmark  \\  \hline

\cite{26262626} & 2025 & RIS & AF and DF & HD & \xmark  \\  \hline

\cite{11142285} & 2025 & STAR-RIS and RIS & DF & HD and FD & \xmark  \\  \hline

This work & 2025 & STAR-RIS and RIS & DF & HD and FD & \cmark  \\  \hline

\end{tabular}
\end{adjustbox}
\end{table*}

\subsection{Motivations}
Since it was shown in \cite{10.1007/978-3-031-28076-4_10, 10598369} and \cite{REF333} that DF relaying outperforms AF relaying, it serves as a more appropriate benchmark for comparison with RIS. Consequently, most comparative studies in the existing literature focus on evaluating RIS-based systems against DF relaying schemes due to their superior performance characteristics. On the other hand, FD relays can simultaneously transmit and receive the signals over the same time-frequency dimension, theoretically doubling the achievable rate compared to traditional HD relays. Similarly, STAR-RISs are functionally analogous to FD systems because of their ability to simultaneously transmit and reflect incoming signals across the entire surface. In other words, both STAR-RISs and FD relays leverage concurrent bidirectional communication over the same time-frequency resources. This makes FD relaying a more relevant benchmark for performance comparisons with STAR-RIS-aided systems. Consequently, the question arises as to whether STAR-RIS-assisted systems can outperform FD relaying, and, if so, under what conditions. Therefore, although \cite{11142285} provides a brief comparison between STAR-RISs and DF relaying in both HD and FD modes in the absence of Es, a comprehensive and fair comparison supported by detailed mathematical analysis in the presence of Es is still required and has not yet been addressed in the existing literature.

\subsection{Contributions}
The main contributions of this paper are summarized as follows:
\begin{itemize}
   \item We present a fundamental analytical comparison of STAR-RIS-assisted systems with SISO, conventional RISs, and DF relaying schemes, including HD and FD modes.
   
   \item We derive closed-form expressions for both the achievable rate and secrecy rate of STAR-RISs in the absence and presence of Es.

   \item We consider transmit power adaptation in both HD- and FD-DF relaying, and the performance is optimized by computing the optimal transmit powers at both S and R. Throughout the paper, the direct link is considered in all schemes, unlike most existing works where it is ignored and the direct S–D signal is assumed to be strongly attenuated.

   \item We present the conditions under which STAR-RIS achieves higher rates than HD-DF relaying and FD-DF relaying. In particular, we quantify the minimum number of STAR-RIS elements required to outperform DF relaying schemes in terms of the achievable rate.

   \item We investigate the effects of key parameters—including transmit power, number of elements, reflection-to-transmission power ratio, element-splitting factor, and deployment positions of RIS/STAR-RIS—on the achievable and secrecy rate performance.

   \item Our results demonstrate that, unlike HD- and FD-DF relaying, the secrecy rate of STAR-RIS systems strongly depends on the number of elements, transmit power, and element-splitting factor. Notably, STAR-RIS systems achieve superior secrecy rates over all benchmark schemes even with a relatively small number of elements.
\end{itemize}

The results presented in this article provide valuable system design guidelines by identifying the appropriate use cases for SISO, HD- and FD-DF relaying modes, as well as RIS- and STAR-RIS-supported transmission.

\subsection{Outline}
The rest of this paper is organized as follows. Sections \ref{sec2} and \ref{sec3} analyze the achievable rate and secrecy rate of the considered system models, respectively. The simulation results are provided in Section \ref{sec4}, and finally, the conclusions are given in Section \ref{sec5}.

\section{Achievable Rate Analysis} \label{sec2}

In this paper, we adopt the system models of SISO and HD-DF relaying as presented in \cite{REF_Emil_3} with the corresponding rate expressions given in \cite[Eq.~(11)]{REF_Emil_3}, and \cite[Eq.~(14)]{REF_Emil_3}, respectively. In the following subsections, we introduce the system models of FD-DF relaying and STAR-RIS, together with their associated achievable rate expressions. We then compare the achievable rates of STAR-RIS with those of HD- and FD-DF relaying, and derive the minimum number of elements required for STAR-RIS to outperform both relaying schemes.

\begin{figure*}[t]
\centering
\subfloat[\label{figure_1a}]{\includegraphics[width=8.4cm,height=4.9cm]{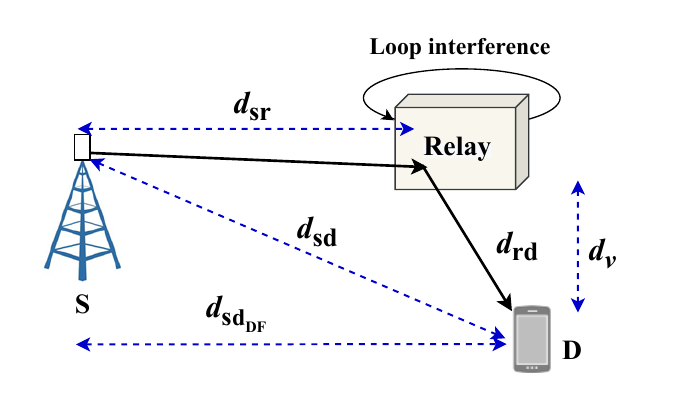}}
\subfloat[\label{figure_1b}]
{\includegraphics[width=9.4cm,height=4.9cm]{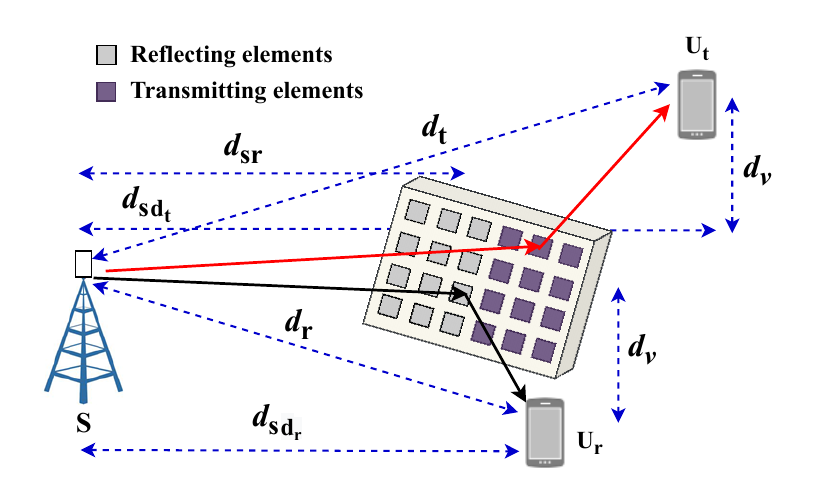}}
\caption{Illustration of (a) a two-antenna relay-supported network and, (b) STAR-RIS-supported network.}
\label{figure_1}
\end{figure*}

\subsection{FD-DF Relaying Transmission}
In FD relaying, as shown in Fig. \ref{figure_1a}, the two-antenna R receives and forwards the wireless signals simultaneously on the same frequency. In this relaying protocol, the communication is degraded by the residual loop-back interference (also known as “self-interference”) from the R transmission mode to the reception mode. At the same instant, S transmits $\sqrt{p_{1}} x$ to R, which is also overheard by D, and the two-antenna R receives $y_{\mathrm{r}}$ and transmits $\sqrt{p_{2}} t$ to D. The received signals at R and D are expressed, respectively, as \cite{5961159}
\begin{equation}\label{equ1}
y_{\mathrm{r}} = h_{\mathrm{sr}} \sqrt{p_{1}} x + h_{\mathrm{li}} \sqrt{p_{2}} t + n_{\mathrm{r}},
\end{equation}
and
\begin{equation}\label{equ2}
y_{\mathrm{d}} = h_{\mathrm{rd}} \sqrt{p_{2}} t + h_{\mathrm{sd}} \sqrt{p_{1}} x + n_{\mathrm{d}},
\end{equation}
where $p_{1}$ and $p_{2}$ denote the transmit powers of S and R, respectively, $x$ is the information signal with unit power, $h_{\mathrm{sr}} \in \mathbb{C}$, $h_{\mathrm{rd}} \in \mathbb{C}$ and $h_{\mathrm{sd}} \in \mathbb{C}$ represent the S-R, R-D and S-D channel coefficients, respectively, $h_{\mathrm{li}} \in \mathbb{C}$ refers to the residual loop interference channel, while $n_{\mathrm{r}} \sim \mathcal{N}\mathcal{C}\left(0, \sigma_{\mathrm{r}}^{2}\right)$ and $n_{\mathrm{d}} \sim \mathcal{N}\mathcal{C}\left(0, \sigma_{\mathrm{d}}^{2}\right)$ are the receiver noises at R and D, respectively, each modeled as a zero-mean complex Gaussian random variable with variances $\sigma_{\mathrm{r}}^{2}$ and $\sigma_{\mathrm{d}}^{2}$.

\begin{lemma} The achievable rate of the FD-DF\footnote{Unlike HD-DF relaying, the data rate at D of the FD-DF relaying does not include the 1/2-prelog penalty, where R is affected by the residual loop interference, while D is impaired by the interference generated by the simultaneous transmission of S and R.} relaying is
\begin{equation}\label{eq_rate_FD_DF}
R_{\mathrm{d}}^{\mathrm{FD-DF}} = \log _{2}\left(1 + \min \left[\frac{p_{1} {\left|h_{\mathrm{sr}}\right|^{2}}}{p_{2} {\left|h_{\mathrm{li}}\right|^{2}} + \sigma_{\mathrm{r}}^{2}}, \frac{p_{2} {\left|h_{\mathrm{rd}}\right|^{2}}}{p_{1} {\left|h_{\mathrm{sd}}\right|^{2}} + \sigma_{\mathrm{d}}^{2}}\right] \right).
\end{equation}
Using the following notations, $\beta_{\mathrm{sr}} = {\left|h_{\mathrm{sr}}\right|^{2}}$, $ \beta_{\mathrm{li}} = {\left|h_{\mathrm{li}}\right|^{2}}$, $\beta_{\mathrm{rd}} = {\left|h_{\mathrm{rd}}\right|^{2}}$, and $ \beta_{\mathrm{sd}} = {\left|h_{\mathrm{sd}}\right|^{2}}$, the rate in \eqref{eq_rate_FD_DF} can be written in terms of the channel gains as
\begin{equation}\label{rate_FD_DF}
R_{\mathrm{d}}^{\mathrm{FD-DF}} = \log _{2}\left(1 + \min \left[\frac{p_{1} \beta_{\mathrm{sr}}}{p_{2} \beta_{\mathrm{li}} + \sigma^{2}}, \frac{p_{2} \beta_{\mathrm{rd}}}{p_{1} \beta_{\mathrm{sd}} + \sigma^{2}}\right] \right),
\end{equation} 
where $\beta_{\mathrm{sr}}$, $\beta_{\mathrm{rd}}$, and $\beta_{\mathrm{sd}}$ are the channel gains of the S-R, R-D, and the direct S-D links, respectively, while $\beta_{\mathrm{li}}$ is the residual loop interference channel gain.
\end{lemma}

\begin{proof}
This is a classical result, as presented in \cite[Eq.~(11)]{5961159}, where $\sigma_{\mathrm{r}}^{2} = \sigma_{\mathrm{d}}^{2} = \sigma^{2}$ is assumed.
\end{proof}

\subsection{STAR-RIS-supported Transmission}
As depicted in Fig. \ref{figure_1b}, the STAR-RIS is considered where the S communicates simultaneously with two single-antenna Ds with the aid of a STAR-RIS that consists of $N_\mathrm{ref}$ passive elements. The two Ds are placed in the transmission ($\text{U}_{\text{t}}$) and reflection ($\text{U}_{\text{r}}$) zones, deployed on opposite sides of the STAR-RIS. The mode switching (MS) protocol is used by STAR-RIS allowing each element to operate in either full transmission mode or full reflection mode. As a result, the STAR-RIS is divided into two parts, the first part comprises $N_\text{t}$ transmitting elements, while the second one contains $N_\text{r}$ reflecting elements \cite{REF_Aldababsa_10}. In this setup, the diagonal reflection and transmission matrices can be defined as $\boldsymbol{\Phi}_{\text{r}, \text{d}} \triangleq \alpha_\text{r} \operatorname {diag} \left(\zeta e^{j \theta_{\text{r},1}}, \ldots, \zeta e^{j \theta_{\text{r},{N}_{\text{r}}}}\right)$ and $\boldsymbol{\Phi}_{\text{t}, \text{d}} \triangleq \alpha_\text{t} \operatorname {diag} \left({\sqrt{1-{\zeta}^2}}e^{j \theta_{\text{t},1}}, \ldots, {\sqrt{1-{\zeta}^2}} e^{j \theta_{\text{t}, {N}_{\text{t}}}}\right)$, respectively, where $\zeta \in$ [0, 1] is the reflection-to-transmission power ratio, $\alpha_\text{r}$ and $\alpha_\text{t} \in$ (0, 1] are the fixed amplitude coefficients, while $\theta_{\text{r},1}, \ldots, \theta_{\text{r}, N_\text{r}}$ and $\theta_{\text{t},1}, \ldots, \theta_{\text{t}, N_\text{t}}$ are the phase-shift parameters that the STAR-RIS is capable of correctly adjusting. The received signals at $\text{U}_{\text{r}}$ and $\text{U}_{\text{t}}$ in the reflection and transmission zones, respectively, are
\begin{equation}\label{ref_received}
y_\mathrm{d}^{\mathrm{ref}}=\left(h_{\mathrm{s {d_{ref}}}} + \mathbf{h}_{\mathrm{s {r_{ref}}}}^{\mathrm{T}} \boldsymbol{\Phi}_{\text{r}, \text{d}} \mathbf{h}_{\mathrm{r d_\mathrm{ref}}}\right) \sqrt{p} x + n_{\text{d}_\mathrm{ref}},
\end{equation}
and
\begin{equation}\label{tra_received}
y_\mathrm{d}^{\mathrm{tra}} = \left(h_{\mathrm{s {d_{tra}}}} + \mathbf{h}_{\mathrm{s {r_{tra}}}}^{\mathrm{T}} \boldsymbol{\Phi}_{\text{t}, \text{d}} \mathbf{h}_{\mathrm{r d_\mathrm{tra}}}\right) \sqrt{p} x + n_{\text{d}_\mathrm{tra}},
\end{equation}
where $p$ denotes the transmit power, $h_{\mathrm{s {d_{ref}}}} \in \mathbb{C}$ and $h_{\mathrm{s {d_{tra}}}} \in \mathbb{C}$
are the channels between the S and Ds in the reflection and transmission zones, respectively, $\mathbf{h}_{\mathrm{s {r_{ref}}}} \in \mathbb{C}^{{N_\text{r}}}$ and $\mathbf{h}_{\mathrm{s {r_{tra}}}} \in \mathbb{C}^{{N_\text{t}}}$ denote the channels between S and the reflection and transmission parts in the STAR-RIS, respectively, $\mathbf{h}_{\mathrm{r d_\mathrm{ref}}} \in \mathbb{C}^{N_\text{r}}$ and $\mathbf{h}_{\mathrm{r d_\mathrm{tra}}} \in \mathbb{C}^{N_\text{t}}$ denote the channel between each part in the STAR-RIS and Ds in the reflection and transmission zones, respectively, whereas $n_{\text{d}_\mathrm{ref}} \sim \mathcal{N}\mathcal{C}\left(0, {\sigma_{\text{d}_\mathrm{ref}}^{2}}\right)$ and $n_{\text{d}_\mathrm{tra}} \sim \mathcal{N}\mathcal{C}\left(0, {\sigma_{\text{d}_\mathrm{tra}}^{2}}\right)$ are the receiver noises at Ds in the reflection and transmission zones, respectively. It is assumed that $\sigma_{\text{d}_\mathrm{ref}}^{2} = \sigma_{\text{d}_\mathrm{tra}}^{2} = \sigma^{2}$.

\begin{lemma} The corresponding rates of $\text{U}_{\text{r}}$ and $\text{U}_{\text{t}}$ are given, respectively, as
\begin{equation}\label{R_D_ref}
R_{\mathrm{d}}^{\mathrm{ref}} = \max _{\theta_{1}, \ldots, \theta_{N_\mathrm{r}}} \log_{2}\left(1 + \frac{p \left|h_{\mathrm{s {d_{ref}}}} + \mathbf{h}_{\mathrm{s {r_{ref}}}}^{\mathrm{T}} \boldsymbol{\Phi}_{\mathrm{r}, \mathrm{d}} \mathbf{h}_{\mathrm{r d_\mathrm{ref}}}\right|^{2}}{\sigma^{2}} \right),
\end{equation}
and
\begin{equation}\label{R_D_tra}
R_{\mathrm{d}}^{\mathrm{tra}} = \max_{\theta_{1}, \ldots, \theta_{N_\mathrm{t}}} \log_{2}\left(1 + \frac{p \left|h_{\mathrm{s {d_{tra}}}} + \mathbf{h}_{\mathrm{s {r_{tra}}}}^{\mathrm{T}}  \boldsymbol{\Phi}_{\mathrm{t}, \mathrm{d}} \mathbf{h}_{\mathrm{r d_\mathrm{tra}}}\right|^{2}}{\sigma^{2}} \right).
\end{equation}
These rates can be written in terms of the channel gains, respectively as
\begin{equation}\label{rate_ref}
\begin{aligned}
R_{\mathrm{d}}^{\mathrm{ref}} =
\log_{2}\left(1 + \frac{p \left(\sqrt{\beta_{\mathrm{s d_\mathrm{ref}}}} + N_\mathrm{r} \alpha_\mathrm{r} \zeta \sqrt{\beta_{\mathrm{sr}} \beta_{\mathrm{r d_\mathrm{ref}}}}\right)^{2}}{\sigma^{2}}\right),
\end{aligned}
\end{equation}
and
\begin{equation}\label{rate_tra}
R_{\mathrm{d}}^{\mathrm{tra}} =
\log_{2} \left(1 + \frac{p \left(\sqrt{\beta_{\mathrm{s d_\mathrm{tra}}}} + N_\mathrm{t} \alpha_\mathrm{t} \sqrt{(1-{\zeta}^2)\beta_{\mathrm{sr}} \beta_{\mathrm{r d_\mathrm{tra}}}}\right)^{2}}{\sigma^{2}} \right),
\end{equation}
where $\beta_{\mathrm{r d_\mathrm{ref}}}$ and $\beta_{\mathrm{r d_\mathrm{tra}}}$ denote the channel gains of the links between the S and Ds in the reflection and transmission zones, respectively.
\end{lemma} 

\begin{proof} 
The achievable rates expressed in \eqref{R_D_ref} and \eqref{R_D_tra} are achieved from the capacity of additive white Gaussian noise channel for given phase-shifting matrices, $\boldsymbol{\Phi}_{\text{r}, \text{d}}$ and $\boldsymbol{\Phi}_{\text{t}, \text{d}}$, respectively. It is worth mentioning that $\mathbf{h}_{\mathrm{s {r_{ref}}}}^{\mathrm{T}} \boldsymbol{\Phi}_{\text{r}, \text{d}} \mathbf{h}_{\mathrm{r d_\mathrm{ref}}}=\alpha_\text{r} \zeta \sum_{n_\text{r}=1}^{N_\text{r}} e^{j \theta_{n_\text{r}}}\left[\mathbf{h}_{\mathrm{s {r_{ref}}}}\right]_{n_\text{r}}\left[\mathbf{h}_{\mathrm{r d_\mathrm{ref}}}\right]_{n_\text{r}}$ and $\mathbf{h}_{\mathrm{s {r_{tra}}}}^{\mathrm{T}} \boldsymbol{\Phi}_{\text{t}, \text{d}} \mathbf{h}_{\mathrm{r d_\mathrm{tra}}} = \alpha_\text{t} {\sqrt{1-{\zeta}^2}} \sum_{n_\text{t}=1}^{N_\text{t}}e^{j \theta_{n_\text{t}}}\left[\mathbf{h}_{\mathrm{s {r_{tra}}}}\right]_{n_\text{t}}\left[\mathbf{h}_{\mathrm{r d_\mathrm{tra}}}\right]_{n_\text{t}}$, 
where $\left[\mathbf{h}_{\mathrm{s {r_{ref}}}}\right]_{n_\text{r}}$ and $\left[\mathbf{h}_{\mathrm{r d_\mathrm{ref}}}\right]_{n_\text{r}}$ are the $n_\text{r}$-th components, while $\left[\mathbf{h}_{\mathrm{s {r_{tra}}}}\right]_{n_\text{t}}$ and $\left[\mathbf{h}_{\mathrm{r d_\mathrm{tra}}}\right]_{n_\text{t}}$ are the $n_\text{t}$-th components. In addition, the phase-shifts are chosen as $\theta_{n_\text{r}}=\arg \left(h_{\mathrm{s d_\mathrm{ref}}}\right)-\arg \left(\left[\mathbf{h}_{\mathrm{s {r_{ref}}}}\right]_{n_\text{r}}\left[\mathbf{h}_{\mathrm{r d_\mathrm{ref}}}\right]_{n_\text{r}}\right)$ and $\theta_{n_\text{t}}=\arg \left(h_{\mathrm{s d_\mathrm{tra}}}\right)-\arg \left(\left[\mathbf{h}_{\mathrm{s {r_{tra}}}}\right]_{n_\text{t}}\left[\mathbf{h}_{\mathrm{r d_\mathrm{tra}}}\right]_{n_\text{t}}\right)$ so that the maximum rate is reached. For simplicity, the following notations, $\sum_{n_\text{r}=1}^{N_\text{r}}\left|\left[\mathbf{h}_{\mathrm{s {r_{ref}}}}\right]_{n_\text{r}}\left[\mathbf{h}_{\mathrm{r d_\mathrm{ref}}}\right]_{n_\text{r}}\right|={N_\text{r}}\sqrt{\beta_{\mathrm{sr}} \beta_{\mathrm{r d_\mathrm{ref}}}}$, $\sum_{n_\text{t}=1}^{N_\text{t}}\left|\left[\mathbf{h}_{\mathrm{s {r_{tra}}}}\right]_{n_\text{t}}\left[\mathbf{h}_{\mathrm{r d_\mathrm{tra}}}\right]_{n_\text{t}}\right|={N_\text{t}}\sqrt{\beta_{\mathrm{sr}} \beta_{\mathrm{r d_\mathrm{tra}}}}$, ${\left|h_{\mathrm{s d_\mathrm{ref}}}\right|^{2}} = \beta_{\mathrm{s d_\mathrm{ref}}}$ and ${\left|h_{\mathrm{s d_\mathrm{tra}}}\right|^{2}} = \beta_{\mathrm{s d_\mathrm{tra}}}$, are used and \eqref{rate_ref} and \eqref{rate_tra} are obtained.
\end{proof}

\subsection{Rate Performance Comparison}
In the following, Proposition \ref{Proposition_1} provides the optimal transmit power for FD-DF relaying. Propositions \ref{Proposition_2} and \ref{Proposition_3} present the conditions under which STAR-RIS achieves higher rates than HD-DF relaying and FD-DF relaying, respectively.

\begin{proposition}\label{Proposition_1}
By selecting $p_1$ and $p_2$ under the constraint $p=\frac{p_{1}+p_{2}}{2}$, the FD-DF relaying rate in \eqref{rate_FD_DF} is maximized by selecting the R transmit power, $p_2$, in terms of the average power $p$, as given in \eqref{p_2}, while $p_1$ is then obtained as $2p-p_2$.
\end{proposition}

\begin{figure*}[t!]
\begin{equation}\label{p_2}
p_2 = \frac{ - {\sigma^2} \left(\beta_{\text{rd}} + \beta_{\text{sr}} \right) - 4 \beta_{\text{sr}} \beta_{\text{sd}} p + \sqrt{{\sigma^4} \left(\beta_{\text{rd}} + \beta_{\text{sr}}\right)^2 
+ 16 \beta_{\text{sr}} \beta_{\text{sd}} \beta_{\text{rd}}\beta_{\text{li}} p^2 
+ 8 {\sigma^2} \beta_{\text{sr}} \beta_{\text{rd}} \left( \beta_{\text{sd}} + \beta_{\text{li}} \right) p }}
{2 \left(\beta_{\text{rd}} \beta_{\text{li}} - \beta_{\text{sr}} \beta_{\text{sd}}\right)}.
\end{equation}
\hrule
\end{figure*}

\begin{proof}
In order to maximize the rate in \eqref{rate_FD_DF}, $p_1$ and $p_2$ are optimally chosen such that $\frac{p_{1} \beta_{\mathrm{sr}}}{p_{2} \beta_{\mathrm{li}} + \sigma^{2}} = \frac{p_{2} \beta_{\mathrm{rd}}}{p_{1} \beta_{\mathrm{sd}} + \sigma^{2}}$ while maintaining the same average power $p$ as in the STAR-RIS case, under the constraint $p=\frac{p_{1}+p_{2}}{2}$. 
Using the well-known quadratic formula, $p_2$ can be mathematically found in terms of $p$, as in \eqref{rate_FD_DF}, while $p_1$ can be found as $2p-p_2$. Thus, the maximum achievable rate of the FD-DF relaying can be obtained as
\begin{equation}\label{rate_FD_DF_cal}
R_{\mathrm{d}}^{\text{FD-DF}} = \log _{2}\left(1 + \frac{p_{2} \beta_{\mathrm{rd}}}{(2p - p_{2}) \beta_{\mathrm{sd}} + \sigma^{2}} \right).
\end{equation}
By substituting \eqref{rate_FD_DF} into \eqref{rate_FD_DF_cal}, the average power $p$ required to achieve a target data rate ${\bar{R}}$ can be found by taking the positive root according to the following formula
\begin{equation}
p = \frac{-b + \sqrt{b^2 - 4ac}}{2a},
\end{equation}
where $a$, $b$ and $c$ are given, respectively, as
\begin{equation}
a = {\left(C_6 C_8 + C_5 C_7 \right)^2 - C_7^2 C_2},
\end{equation}
\begin{equation}
b = {2\left(C_6 C_8 + C_5 C_7\right) \left(C_6 C_9 + C_4 C_7\right) - C_7^2 C_3},
\end{equation}
and
\begin{equation}
c = \left(C_6 C_9 + C_4 C_7 \right)^2 - C_7^2 C_1,
\end{equation}
where $C_1 = \big(\sigma^2(\beta_{\text{rd}} + \beta_{\text{sr}}) \big)^2$, $C_2 = 16\beta_{\text{sr}}\beta_{\text{sd}}\beta_{\text{rd}}\beta_{\text{li}}$, $C_3 = 8\beta_{\text{sr}}\beta_{\text{rd}}\sigma^2 \big(\beta_{\text{sd}} + \beta_{\text{li}}\big)$, $C_4 = \sigma^2 \big(\beta_{\text{rd}} + \beta_{\text{sr}} \big)$, $C_5 = 4\beta_{\text{sr}}\beta_{\text{sd}}$, $C_6 = 2 \big(\beta_{\text{rd}}\beta_{\text{li}} - \beta_{\text{sr}}\beta_{\text{sd}} \big)$, $C_7 = \big(2^{\bar{R}}-1 \big)\beta_{\text{sd}} + \beta_{\text{rd}}$, $C_8 = 2 \big(2^{\bar{R}}-1 \big)\beta_{\text{sd}}$ and $C_9 = \big(2^{\bar{R}}-1 \big)\sigma^2$.
\end{proof}

\begin{proposition}\label{Proposition_2}
The highest rate can be provided by the STAR-RIS-supported transmission at Ds in the reflection and transmission zones for any $N_\mathrm{r} \geq 1$ and $N_\mathrm{t} \geq 1$ if $\beta_{\mathrm{s d_{ref}}}>\beta_{\mathrm{sr}}$ and $\beta_{\mathrm{s d_{tra}}}>\beta_{\mathrm{sr}}$, respectively. On the other hand, the STAR-RIS-supported network gives higher rate at Ds in the reflection and transmission zones than the rate of the HD-DF relaying when $\beta_{\mathrm{s d_{ref}}} \leq \beta_{\mathrm{sr}}$ and $\beta_{\mathrm{s d_{tra}}} \leq \beta_{\mathrm{sr}}$, respectively, if and only if
\begin{equation}\label{Nr_max}
N_\mathrm{r}>\frac{\sqrt{\left(\sqrt{1+\frac{2 p \beta_{\mathrm{sr}} \beta_{\mathrm{rd}}}{\left(\beta_{\mathrm{sr}}+\beta_{\mathrm{rd}}-\beta_{\mathrm{sd}}\right) \sigma^{2}}}-1\right) \frac{\sigma^{2}}{p}}-\sqrt{\beta_{\mathrm{s d_\mathrm{ref}}}}}{\alpha_\text{r} \zeta \sqrt{{\beta_{\mathrm{sr}}}{\beta_{\mathrm{r d_\mathrm{ref}}}} }},
\end{equation}
and
\begin{equation}\label{Nt_max}
N_\mathrm{t}>\frac{\sqrt{\left(\sqrt{1+\frac{2 p \beta_{\mathrm{sr}} \beta_{\mathrm{rd}}}{\left(\beta_{\mathrm{sr}}+\beta_{\mathrm{rd}}-\beta_{\mathrm{sd}}\right) \sigma^{2}}}-1\right) \frac{\sigma^{2}}{p}}-\sqrt{\beta_{\mathrm{s d_\mathrm{tra}}}}}{\alpha_\text{t} \sqrt{1-{\zeta}^2} \sqrt{{\beta_{\mathrm{sr}}}{\beta_{\mathrm{r d_\mathrm{tra}}}} }}.
\end{equation}
\end{proposition}

\begin{proof}
Since $R_{\mathrm{ref}}(N_\text{r})>R_{\mathrm{SISO}}$ and $R_{\mathrm{tra}}(N_\text{t})>R_{\mathrm{SISO}}$ for $N_\text{r} \geq 1$ and $N_\text{t} \geq 1$, respectively, the STAR-RIS provides the highest data rate only if $R_{\mathrm{ref}}(N_\text{r})>$ $R_{\text{HD-DF}}$ and $R_{\mathrm{tra}}(N_\text{t})>$ $R_{\text{HD-DF}}$. These, however, occur only when $\beta_{\mathrm{s d_\mathrm{ref}}}>\beta_{\mathrm{sr}}$ and $\beta_{\mathrm{s d_\mathrm{tra}}}>\beta_{\mathrm{sr}}$ since $R_{\mathrm{SISO}}>R_{\mathrm{DF}}$. On the other hand, if $\beta_{\mathrm{s d_\mathrm{ref}}} \leq \beta_{\mathrm{sr}}$ and $\beta_{\mathrm{s d_\mathrm{tra}}} \leq \beta_{\mathrm{sr}}$, the inequalities $R_{\mathrm{ref}}(N_\text{r})>$ $R_{\text{HD-DF}}$ and $R_{\mathrm{tra}}(N_\text{t})>$ $R_{\text{HD-DF}}$ can be simplified, respectively, to \eqref{Nr_max} and \eqref{Nt_max} by utilizing 
\cite[Eq.~(14)]{REF_Emil_3}, \eqref{rate_ref} and \eqref{rate_tra}.
\end{proof}

\begin{proposition}\label{Proposition_3}
The STAR-RIS-supported system with Ds in the reflection and transmission zones, and the conventional RIS system provide higher data rate than the FD-DF relaying, respectively, if and only if
\begin{equation}\label{eq19}
N_\mathrm{r} >
\frac{
	\sqrt{
		\frac{
			p_{2}\,\sigma^{2}\,\beta_{\mathrm{rd}}
		}{
			p\Big( (2p - p_{2})\,\beta_{\mathrm{sd}} + \sigma^{2} \Big)
		}
	}
	-
	\sqrt{\beta_{\mathrm{sd,ref}}}
}{
	\alpha_\mathrm{r}\,\zeta\,\sqrt{ \beta_{\mathrm{sr}}\,\beta_{\mathrm{rd,ref}} }
},
\end{equation}
\begin{equation}\label{eq20}
N_\mathrm{t} >
\frac{
	\sqrt{
		\frac{
			p_{2}\,\sigma^{2}\,\beta_{\mathrm{rd}}
		}{
			p\Big( (2p - p_{2})\,\beta_{\mathrm{sd}} + \sigma^{2} \Big)
		}
	}
	-
	\sqrt{\beta_{\mathrm{sd,tra}}}
}{
	\alpha_{\mathrm{t}} 
	\sqrt{1-\zeta^{2}}
	\sqrt{\beta_{\mathrm{sr}}\,\beta_{\mathrm{rd,tra}}}
}
\end{equation}
and
\begin{equation}\label{eq21}
N_\mathrm{ref} >
\frac{
	\sqrt{
		\frac{
			p_{2}\,\sigma^{2}\,\beta_{\mathrm{rd}}
		}{
			p\Big( (2p - p_{2})\,\beta_{\mathrm{sd}} + \sigma^{2} \Big)
		}
	}
	-
	\sqrt{\beta_{\mathrm{sd}}}
}{
	\alpha \sqrt{ \beta_{\mathrm{sr}}\,\beta_{\mathrm{rd}} }
}.
\end{equation}
\end{proposition}

\begin{proof}
 For $\beta_{\mathrm{s d_\mathrm{ref}}} \leq \beta_{\mathrm{sr}}$, $\beta_{\mathrm{s d_\mathrm{tra}}} \leq \beta_{\mathrm{sr}}$ and $\beta_{\mathrm{sd}} \leq \beta_{\mathrm{sr}}$, the inequalities $R_{\mathrm{ref}}(N_\text{r})>$ $R_{\text{FD-DF}}$, $R_{\mathrm{tra}}(N_\text{t})>$ $R_{\text{FD-DF}}$ and $R_{\mathrm{RIS}}(N_\mathrm{ref})>$ $R_{\text{FD-DF}}$ can be simplified, respectively, to \eqref{eq19}, \eqref{eq20} and \eqref{eq21} by utilizing 
\cite[Eq.~(12)]{REF_Emil_3}, \eqref{rate_ref}, \eqref{rate_tra} and \eqref{rate_FD_DF_cal}. 
\end{proof}

\section{Secrecy Rate Analysis}\label{sec3}

In this section, we analyze the performance of SISO, HD-DF relaying, FD-DF relaying, RIS-assisted, and STAR-RIS-assisted systems in the presence of Es. In particular, we derive the corresponding achievable secrecy rate expressions for these systems as follows.

\subsection{Secrecy rate of SISO transmission}
Considering a SISO system with a deterministic flat-fading channel denoted by $h_{\mathrm{sd}} \in \mathbb{C}$, the received signal at D is expressed as
\begin{equation}
y_{\mathrm{d}} = h_{\mathrm{sd}} \sqrt{p} x + n_{\mathrm{d}},
\end{equation}
where $n_{\mathrm{d}}$ is the receiver noise with zero mean and variance of $\sigma_{\mathrm{d}}^2$, $n_{\mathrm{d}} \sim \mathcal{N}\mathcal{C}\left(0, \sigma_{\mathrm{d}}^{2}\right)$. 

In the presence of an E in the SISO system, the received signal at E is expressed as
\begin{equation}
y_{\mathrm{e}} = h_{\mathrm{se}} \sqrt{p} x + n_{\mathrm{e}},
\end{equation}
where $h_{\mathrm{se}} \in \mathbb{C}$ is the S-E channel, and $n_{\mathrm{e}} \sim \mathcal{N}\mathcal{C}\left(0, \sigma_{\mathrm{e}}^{2}\right)$ is the receiver noise at E which has a zero mean and variance of $\sigma_{\mathrm{e}}^{2}$.

\begin{proposition}\label{Proposition_4}
The secrecy rate of the SISO system, denoted by $R_{\mathrm{sec}}^{\mathrm{SISO}}$, is given by
\begin{equation}\label{secrecy_SISO}
R_{\mathrm{sec}}^{\mathrm{SISO}} = \max \Bigg[0, \log_{2}\left(\frac{\sigma^{2} + p \beta_{\mathrm{sd}}}{\sigma^{2} + {p \beta_{\mathrm{se}}}}\right) \Bigg].
\end{equation}
\end{proposition}

\begin{proof}
The secrecy rate of a SISO system can generally be expressed as  $ R_{\mathrm{sec}}^{\text{SISO}} = \max\Big[0,R_{\mathrm{d}}^{\text{SISO}} - R_{\mathrm{e}}^{\text{SISO}}\Big]$. 
Assuming that the antenna gains are included in the channels for notational simplicity, the channel capacity of the SISO link at D is expressed as \cite{REF_Emil_3}
\begin{equation}
R_{\mathrm{d}}^{\text{SISO}} = \log_{2}\left(1 + \frac{p \beta_{\mathrm{sd}}}{\sigma_{\mathrm{d}}^{2}}\right),
\end{equation}
where the channel gain $\beta_{\mathrm{sd}} = {\left|h_{\mathrm{sd}}\right|^{2}}$ is considered. Similarly, the channel capacity of the SISO links at E is expressed as
\begin{equation}
R_{\mathrm{e}}^{\text{SISO}} = \log_{2}\left(1 + \frac{p \beta_{\mathrm{se}}}{\sigma_{\mathrm{e}}^{2}}\right),
\end{equation}
where $\beta_{\mathrm{se}} = {\left|h_{\mathrm{se}}\right|^{2}}$ is the channel gain of the direct S-E link. After simplification, and assuming $\sigma_{\mathrm{e}}^{2} = \sigma_{\mathrm{d}}^{2} = \sigma^{2}$, the
aforementioned secrecy rate expression in \eqref{secrecy_SISO} can be readily obtained.
\end{proof}

\subsection{Secrecy rate of HD-DF relaying}

Consider the traditional repetition-coded HD-DF relaying scheme in which the transmission is divided into two phases of equal size. In the first phase, S transmits $\sqrt{p_{1}} \, x$ to both D and R, where the received signals at D and R are presented as
\begin{equation}
y_{1 \mathrm{d}}=h_{\mathrm{sd}} \sqrt{p_{1}} x + n_{1 \mathrm{d}},
\end{equation}
and
\begin{equation}
y_{1 \mathrm{r}}=h_{\mathrm{sr}} \sqrt{p_{1}} x + n_{1 \mathrm{r}},
\end{equation}
respectively, where $n_{1 \mathrm{d}}, \sim \mathcal{N}\mathcal{C}\left(0, \sigma_{\mathrm{d}}^{2}\right)$ and $n_{1 \mathrm{r}}, \sim \mathcal{N}\mathcal{C}\left(0, \sigma_{\mathrm{r}}^{2}\right)$ are the receiver noises at D and R, respectively.
The E still overhears the transmission from S, and therefore, the received signal at E is expressed as
\begin{equation}
y_{1 \mathrm{e}}=h_{\mathrm{se}} \sqrt{p_{1}} x + n_{1 \mathrm{e}},
\end{equation}
where $n_{1 \mathrm{e}}, \sim \mathcal{N}\mathcal{C}\left(0, \sigma_{\mathrm{e}}^{2}\right)$ is the receiver noise at E.

In the second phase, R transmits $\sqrt{p_{2}} \, x$ to D, and the received signal at D is expressed as
\begin{equation}
y_{2 d}=h_{\mathrm{rd}} \sqrt{p_{2}} x + n_{2 \mathrm{d}},
\end{equation}
while E receives from R the following signal
\begin{equation}
y_{2 \mathrm{e}}=h_{\mathrm{re}} \sqrt{p_{2}} x + n_{2 \mathrm{e}},
\end{equation}
where $h_{\mathrm{re}} \in \mathbb{C}$ is the R-E channel. Without loss of generality, this paper assumes $\sigma_{\mathrm{e}}^{2} = \sigma_{\mathrm{r}}^{2} = \sigma_{\mathrm{d}}^{2} = \sigma^{2}$. It also considers an informed E that decodes the information from both S and R.

\begin{proposition}\label{Propostion_5}
The secrecy rate of the HD-DF relaying system, denoted by $R_{\mathrm{sec}}^{\mathrm{HD-DF}}$, is given by
\begin{multline}\label{rate_secrecy_HD-DF}
R_{\mathrm{sec}}^{\mathrm{HD-DF}} = \max \Bigg[0, \frac{1}{2} \log_{2}\left(1 + \frac{2 p \beta_{\mathrm{sr}} \beta_{\mathrm{rd}}}{\left(\beta_{\mathrm{sr}}+\beta_{\mathrm{rd}}-\beta_{\mathrm{sd}}\right) \sigma^{2}} \right) \\
- \frac{1}{2} \log_{2}\left(1 + \frac{p_{1} \beta_{\mathrm{se}} + p_{2} \beta_{\mathrm{re}}}{\sigma^{2}} \right) \Bigg],
\end{multline}
where $\beta_{\mathrm{se}} = {\left|h_{\mathrm{se}}\right|^{2}}$ and $\beta_{\mathrm{re}} = {\left|h_{\mathrm{re}}\right|^{2}}$ are the channel gains of the S-E and R-E links, respectively, $p_{1}=\frac{2 p \beta_{\mathrm{rd}}}{\beta_{\mathrm{sr}}+\beta_{\mathrm{rd}}-\beta_{\mathrm{sd}}}$ and $p_{2}=\frac{2 p\left(\beta_{\mathrm{sr}}-\beta_{\mathrm{sd}}\right)}{\beta_{\mathrm{sr}}+\beta_{\mathrm{rd}}-\beta_{\mathrm{sd}}}$ are optimally selected to have similar average power, $p$, as when utilizing the STAR-RIS and RIS under the constraint $p=\frac{p_{1}+p_{2}}{2}$, and to maximize the achievable rate of the HD-DF relaying.
\end{proposition}

\begin{proof}
The secrecy rate of the HD-DF relaying in \eqref{rate_secrecy_HD-DF} is generally found as $R_{\mathrm{sec}}^{\text{HD-DF}} = \max\left[0, R_{\mathrm{d}}^{\text{HD-DF}} - R_{\mathrm{e}}^{\text{HD-DF}}\right]$. In this setup, $R_{\mathrm{d}}^{\text{HD-DF}}$ is calculated as \cite{REF_Emil_3}
\begin{equation}\label{rate_HD_DF}
R_{\mathrm{d}}^{\text{HD-DF}}=\frac{1}{2} \log _{2}\left(1+\frac{2 p \beta_{\mathrm{sr}} \beta_{\mathrm{rd}}}{\left(\beta_{\mathrm{sr}}+\beta_{\mathrm{rd}}-\beta_{\mathrm{sd}}\right) \sigma^{2}}\right),
\end{equation}
while the rate at E, denoted by $R_{\mathrm{e}}^{\text{HD-DF}}$, can be formulated as
\begin{equation}
R_{\mathrm{e}}^{\text{HD-DF}} = \frac{1}{2} \log _{2}\left(1 + \max \left[\frac{p_{1} {\left|h_{\mathrm{se}}\right|^{2}}}{\sigma_{\mathrm{e}}^{2}}, \frac{p_{2} {\left|h_{\mathrm{re}}\right|^{2}}}{\sigma_{\mathrm{e}}^{2}}\right] \right),
\end{equation}
which can be written in terms of the channel gains as
\begin{equation}\label{rate_e_HD_DF}
R_{\mathrm{e}}^{\text{HD-DF}} = \frac{1}{2} \log_{2}\left(1 + \max \left[\frac{p_{1} \beta_{\mathrm{se}}}{\sigma^{2}}, \frac{p_{2} \beta_{\mathrm{re}}}{\sigma^{2}} \right]\right).
\end{equation}
In this case, assuming an informed E with the knowledge of the R’s operation is deployed, the rate can be expressed as
\begin{equation}\label{rate_e_HD_DF_final}
R_{\mathrm{e}}^{\text{HD-DF}} = \frac{1}{2} \log_{2}\left(1 + \frac{p_{1} \beta_{\mathrm{se}} + p_{2} \beta_{\mathrm{re}}}{\sigma^{2}} \right).
\end{equation}
Substituting \eqref{rate_HD_DF} and \eqref{rate_e_HD_DF_final} into the general expression of $R_{\mathrm{sec}}^{\text{HD-DF}}$ yields the secrecy rate expression of the HD-DF relaying system in \eqref{rate_secrecy_HD-DF}.
\end{proof}

\subsection{Secrecy rate of FD-DF relaying}

In the FD-DF relaying system, the received signals at R and D are given in \eqref{equ1} and \eqref{equ2}, respectively, while the received signal at E, as shown in Fig. \ref{figure_2a}, is given by
\begin{equation}
y_{\mathrm{e}} = h_{\mathrm{se}} \sqrt{p_{1}} x + h_{\mathrm{re}} \sqrt{p_{2}} t + n_{\mathrm{e}}.
\end{equation}

\begin{proposition}\label{Proposition_6}
For an informed E with knowledge of the R’s operation, the secrecy rate of the FD-DF relaying system, denoted by $R_{\mathrm{sec}}^{\mathrm{FD-DF}}$, is given by 
\begin{multline}\label{rate_secrecy_FD-DF}
R_{\mathrm{sec}}^{\mathrm{FD-DF}} = \max \Bigg[0, \log_{2}\left(1 + \frac{p_{2} \beta_{\mathrm{rd}}}{(2p - p_{2}) \beta_{\mathrm{sd}} + \sigma^{2}} \right) \\
- \log_{2}\left(1 + \frac{p_{1} \beta_{\mathrm{se}}}{\sigma^{2}} + \frac{p_{2} \beta_{\mathrm{re}}}{p_{1} \beta_{\mathrm{se}} + \sigma^{2}} \right) \Bigg].
\end{multline}
Here, the same power values as those derived in Section II are employed.
\end{proposition}

\begin{proof}
The secrecy rate of the FD-DF relaying is generally expressed as $R_{\mathrm{sec}}^{\text{FD-DF}} = \max\big[0, R_{\mathrm{d}}^{\text{FD-DF}} - R_{\mathrm{e}}^{\text{FD-DF}}\big]$. 
In this setup, $R_{\mathrm{d}}^{\text{FD-DF}}$ is expressed in \eqref{rate_FD_DF_cal}, while the rate at E, denoted by $R_{\mathrm{e}}^{\text{FD-DF}}$, can be formulated as
\begin{equation}
R_{\mathrm{e}}^{\text{FD-DF}} = \log _{2}\left(1 + \max \left[\frac{p_{1} {\left|h_{\mathrm{se}}\right|^{2}}}{\sigma_{\mathrm{e}}^{2}}, \frac{p_{2} {\left|h_{\mathrm{re}}\right|^{2}}}{p_{1} {\left|h_{\mathrm{se}}\right|^{2}} + \sigma_{\mathrm{e}}^{2}}\right] \right),
\end{equation}
which can be written in terms of the channel gains as
\begin{equation}
R_{\mathrm{e}}^{\text{FD-DF}} = \log _{2}\left(1 + \max \left[\frac{p_{1} \beta_{\mathrm{se}}}{\sigma^{2}}, \frac{p_{2} \beta_{\mathrm{re}}}{p_{1} \beta_{\mathrm{se}} + \sigma^{2}}\right] \right).
\end{equation}
Given that E knows R's operation, the rate is expressed as
\begin{equation}\label{rate_e_FD-DF}
R_{\mathrm{e}}^{\text{FD-DF}} = \log_{2}\left(1 + \frac{p_{1} \beta_{\mathrm{se}}}{\sigma^{2}} + \frac{p_{2} \beta_{\mathrm{re}}}{p_{1} \beta_{\mathrm{se}} + \sigma^{2}} \right).
\end{equation}
By substituting \eqref{rate_FD_DF_cal} and \eqref{rate_e_FD-DF} into the general expression of $R_{\mathrm{sec}}^{\text{FD-DF}}$, the secrecy rate expression of the FD-DF relaying system in \eqref{rate_secrecy_FD-DF} is obtained.
\end{proof}

\subsection{Secrecy rate of RIS-supported transmission}

In RIS-supported secure transmission, the system consists of an RIS unit with $N_\mathrm{ref}$ passive reflecting elements, a single-antenna S, two single-antenna Ds, one in the reflection zone and one in the transmission zone, a single-antenna E in the reflection zone and another single-antenna E in the transmission zone. The received signal at the legitimate user is expressed as
\begin{equation}
y_\mathrm{d} = \left(h_{\mathrm{sd}} + \mathbf{h}_{\mathrm{sr}}^{\mathrm{T}} \boldsymbol{\Phi} \mathbf{h}_{\mathrm{rd}}\right) \sqrt{p} x + n_{\text{d}},
\end{equation}
where $\mathbf{h}_{\mathrm{sr}} \in \mathbb{C}^{N_\mathrm{ref}}$ and $\mathbf{h}_{\mathrm{rd}} \in \mathbb{C}^{N_\mathrm{ref}}$ denote the channel between S and the RIS, and the channel between the RIS and D, respectively. The received signal at E is
\begin{equation}
y_\mathrm{e} = \sqrt{p} \, h_{\mathrm{se}} x + \sqrt{p} \, \mathbf{h}_{\mathrm{sr}}^{\mathrm{T}} \boldsymbol{\Phi} \mathbf{h}_{\mathrm{re}} x + n_{\text{e}},
\end{equation}
where $\mathbf{h}_{\mathrm{re}} \in \mathbb{C}^{N_\mathrm{ref}}$ denotes the channel between RIS and E. The diagonal phase-shifting matrix that fully captures the characteristics of the RIS can be defined as $\boldsymbol{\Phi} \triangleq \alpha \operatorname {diag} \left(e^{j \theta_{1}}, \ldots, e^{j \theta_{{N}_{\mathrm{ref}}}}\right)$, where $\alpha$ is the fixed amplitude coefficient.

\begin{proposition}\label{Proposition_7}
The secrecy rate of the RIS-supported system, denoted by $R_{\mathrm{sec}}^{\mathrm{RIS}}$, is given by
\begin{multline}\label{rate_secrecy_RIS}
R_{\mathrm{sec}}^{\mathrm{RIS}} = \max\Bigg[0, \log_{2}\left(1 + \frac{p \left(\sqrt{\beta_{\mathrm{sd}}} + N_\mathrm{ref} \alpha \sqrt{\beta_{\mathrm{sr}} \beta_{\mathrm{rd}}}\right)^{2}}{\sigma^{2}}\right) \\
- \log_{2}\left(1 + \frac{p {\beta_{\mathrm{se}}}}
{p N_\mathrm{ref}^2 \alpha^2 \beta_{\mathrm{sr}} \beta_{\mathrm{re}} + \sigma^{2}}\right) \Bigg].
\end{multline}
\end{proposition}

\begin{proof}
The secrecy rate of the RIS in \eqref{rate_secrecy_RIS} can be generally found as $R_{\mathrm{sec}}^{\text{RIS}} =\max\left [0, R_{\mathrm{d}}^{\text{RIS}} - R_{\mathrm{e}}^{\text{RIS}}\right]$. In this setup, the achievable rate at D, denoted by $R_{\mathrm{d}}^{\text{RIS}}$, is expressed as
\cite{REF_Emil_3}
\begin{multline}
R_{\mathrm{d}}^{\text{RIS}} = \max _{\theta_{1}, \ldots, \theta_{N_\mathrm{ref}}} \log _{2} \left( 1 + \frac{p}{\sigma^{2}}\left|h_{\mathrm{sd}}+\mathbf{h}_{\mathrm{sr}}^{\mathrm{T}} \boldsymbol{\Phi} \mathbf{h}_{\mathrm{rd}}\right|^{2}\right)\\
= \log_{2} \left( 1 + \frac{p}{\sigma^{2}}\bigg( \left| h_{\mathrm{sd}}\right| + \alpha \sum_{n=1}^{N_\mathrm{ref}}  e^{j \theta_{n}}\left[\mathbf{h}_{\mathrm{sr}}\right]_{n}\left[\mathbf{h}_{\mathrm{rd}}\right]_{n} \bigg)^{2} \right).
\end{multline}
This rate can be expressed in terms of the channel gains as
\begin{equation}\label{rate_RIS}
\begin{aligned}
R_{\mathrm{d}}^{\text{RIS}} =
\log_{2}\left(1+\frac{p}{\sigma^{2}}\left(\sqrt{\beta_{\mathrm{sd}}}+ N_\mathrm{ref} \alpha \sqrt{\beta_{\mathrm{sr}}\beta_{\mathrm{rd}}}\right)^{2}\right).
\end{aligned}
\end{equation}
This expression is achieved for a given phase-shifting matrix $\boldsymbol{\Phi}$, where the term $\mathbf{h}_{\mathrm{sr}}^{\mathrm{T}} \boldsymbol{\Phi} \mathbf{h}_{\mathrm{rd}}$ is simplified to $\alpha \sum_{n=1}^{N_\mathrm{ref}} e^{j \theta_{n}}\left[\mathbf{h}_{\mathrm{sr}}\right]_{n}\left[\mathbf{h}_{\mathrm{rd}}\right]_{n}$, where $\left[\mathbf{h}_{\mathrm{sr}}\right]_{n}$ and $\left[\mathbf{h}_{\mathrm{rd}}\right]_{n}$ are the $n$-th components. When the phase-shifts are chosen as $\theta_{n}=\arg \left(h_{\mathrm{sd}}\right)-\arg \left(\left[\mathbf{h}_{\mathrm{sr}}\right]_{n}\left[\mathbf{h}_{\mathrm{rd}}\right]_{n}\right)$, the maximum rate is reached so that all terms in the summation have the same phase as $h_{\mathrm{sd}}$. In addition,  the notation $\sum_{n=1}^{N_\mathrm{ref}}\left|\left[\mathbf{h}_{\mathrm{sr}}\right]_{n}\left[\mathbf{h}_{\mathrm{rd}}\right]_{n}\right|={N_\mathrm{ref}}\sqrt{\beta_{\mathrm{sr}} \beta_{\mathrm{rd}}}$ is adopted to derive the achievable rate in \eqref{rate_RIS}.

The rate at E, denoted by $R_{\mathrm{e}}^{\text{RIS}}$, can be expressed as
\begin{multline}
R_{\mathrm{e}}^{\text{RIS}} = \min_{\theta_{1}, \ldots, \theta_{N_\mathrm{ref}}} \log_{2} \left( 1 + \frac{p \left|h_{\mathrm{se}} \right|^{2}}
{p \left|\mathbf{h}_{\mathrm{sr}}^{\mathrm{T}} \boldsymbol{\Phi} \mathbf{h}_{\mathrm{re}}\right|^{2} + \sigma^{2}} \right)\\
= \log_{2} \left( 1 + \frac{p \left|h_{\mathrm{se}} \right|^{2} }
{p \big( \alpha \sum_{n=1}^{N_\mathrm{ref}}  e^{j \theta_{n}}\left[\mathbf{h}_{\mathrm{sr}}\right]_{n}\left[\mathbf{h}_{\mathrm{re}}\right]_{n} \big)^{2} + \sigma^{2}} \right),
\end{multline}
where the signal from the RIS is treated as interference by the E. Therefore, the rate at E can be expressed in terms of the channel gains as
\begin{equation}\label{rate_e_RIS}
R_{\mathrm{e}}^{\text{RIS}} = \log_{2}\left(1 + \frac{p {\beta_{\mathrm{se}}}}
{p N_\mathrm{ref}^2 \alpha^2 \beta_{\mathrm{sr}} \beta_{\mathrm{re}} + \sigma^{2}} \right),
\end{equation}
where the term $\mathbf{h}_{\mathrm{sr}}^{\mathrm{T}} \boldsymbol{\Phi} \mathbf{h}_{\mathrm{re}}$ is simplified to $\alpha \sum_{n=1}^{N_\mathrm{ref}} e^{j \theta_{n}}\left[\mathbf{h}_{\mathrm{sr}}\right]_{n}\left[\mathbf{h}_{\mathrm{re}}\right]_{n}$. Here, $\left[\mathbf{h}_{\mathrm{re}}\right]_{n}$ is the $n$-th component. When the $n$-th phase-shifts are chosen as $\theta_{n}=\arg \left(h_{\mathrm{se}}\right)-\arg \left(\left[\mathbf{h}_{\mathrm{sr}}\right]_{n}\left[\mathbf{h}_{\mathrm{re}}\right]_{n}\right)$, the optimal rate is reached so that all terms in the summation have the same phase as $h_{\mathrm{se}}$. In addition,  the notation $\sum_{n=1}^{N_\mathrm{ref}}\left|\left[\mathbf{h}_{\mathrm{sr}}\right]_{n}\left[\mathbf{h}_{\mathrm{re}}\right]_{n}\right|={N_\mathrm{ref}}\sqrt{\beta_{\mathrm{sr}} \beta_{\mathrm{re}}}$ is adopted to derive the achievable rate in \eqref{rate_e_RIS}. By substituting \eqref{rate_RIS} and \eqref{rate_e_RIS} into the general expression of $R_{\mathrm{sec}}^{\text{RIS}}$, the secrecy rate expression of the RIS-supported system in \eqref{rate_secrecy_RIS} is obtained.
\end{proof}

\begin{figure*}[t]
\centering
\subfloat[\label{figure_2a}]{\includegraphics[width=9.3cm,height=5.4cm]{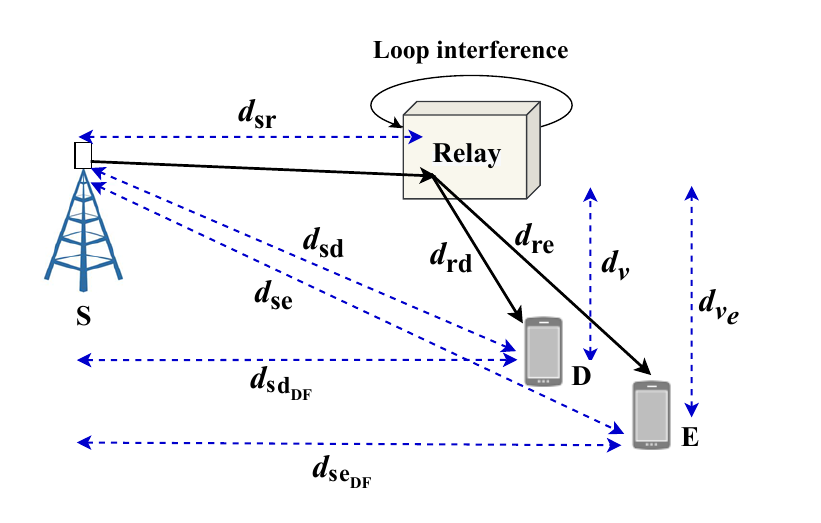}}
\subfloat[\label{figure_2b}]
{\includegraphics[width=9cm,height=5.4cm]{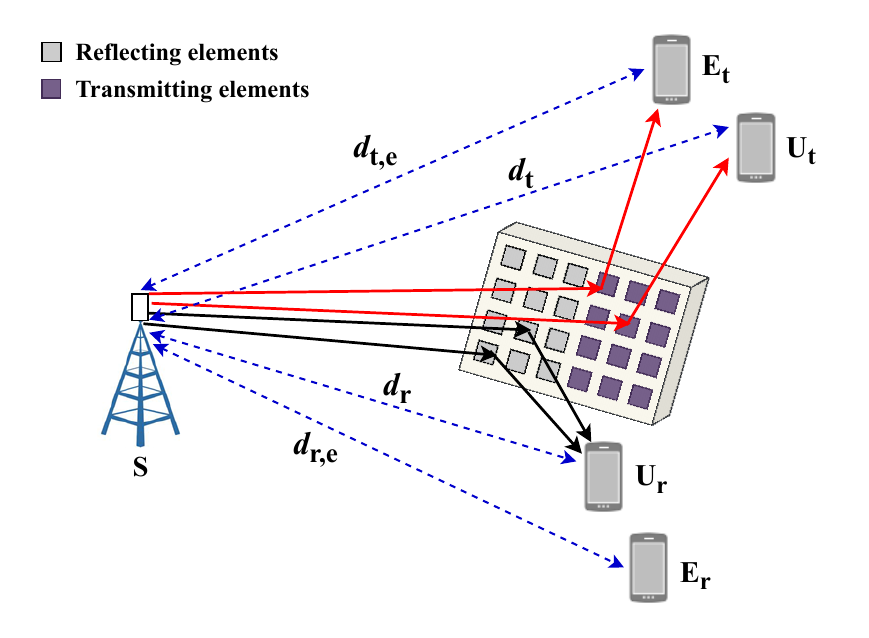}}
\caption{Illustration of (a) a two-antenna relay-supported network with an E and, (b) STAR-RIS-supported network with Es.}
\label{figure_2}
\end{figure*}

\subsection{Secrecy rate of STAR-RIS-supported transmission}

In this subsection, a STAR-RIS-assisted secure communication system, as shown in Fig. \ref{figure_2b}, is considered. This system consists of a single-antenna S, two single-antenna Ds, one in the reflection zone and one in the transmission zone, and an E in the reflection zone and another E in the transmission zone, with the aid of a STAR-RIS that consists of $N_\mathrm{ref}$ passive elements. The STAR-RIS is divided into two parts with $N_\text{t}$ transmitting elements and $N_\text{r}$ reflecting elements. The received signals at the legitimate users in the reflection and transmission zones are expressed as in \eqref{ref_received} and \eqref{tra_received}, respectively, while the received signals at the Es in the reflection and transmission zones are
\begin{equation}
y_\mathrm{e}^{\mathrm{ref}} = \sqrt{p} \, h_{\mathrm{s {e_{ref}}}} x + \sqrt{p} \, \mathbf{h}_{\mathrm{s {r_{ref}}}}^{\mathrm{T}} \boldsymbol{\Phi}_{\text{r}, \text{e}} \mathbf{h}_{\mathrm{r e_\mathrm{ref}}} x + n_{\text{e}_\mathrm{ref}},
\end{equation}
and
\begin{equation}
y_\mathrm{e}^{\mathrm{tra}} = \sqrt{p} \, h_{\mathrm{s {e_{tra}}}} x + \sqrt{p} \, \mathbf{h}_{\mathrm{s {r_{tra}}}}^{\mathrm{T}} \boldsymbol{\Phi}_{\text{t}, \text{e}} \mathbf{h}_{\mathrm{r e_\mathrm{tra}}} x + n_{\text{e}_\mathrm{tra}},
\end{equation}
respectively, where $h_{\mathrm{s {e_{ref}}}} \in \mathbb{C}$ and $h_{\mathrm{s {e_{tra}}}} \in \mathbb{C}$
are the channels between S and Es in the reflection and transmission zones, respectively, $\mathbf{h}_{\mathrm{r e_\mathrm{ref}}} \in \mathbb{C}^{N_\text{r}}$ and $\mathbf{h}_{\mathrm{r e_\mathrm{tra}}} \in \mathbb{C}^{N_\text{t}}$ denote the channel between each part in the STAR-RIS and Es in the reflection and transmission zones, respectively, while $n_{\text{e}_\mathrm{ref}} \sim \mathcal{N}\mathcal{C}\left(0, {\sigma_{\text{e}_\mathrm{ref}}^{2}}\right)$ and $n_{\text{e}_\mathrm{tra}} \sim \mathcal{N}\mathcal{C}\left(0, {\sigma_{\text{e}_\mathrm{tra}}^{2}}\right)$ are the receiver noises at Es in the reflection and transmission zones, respectively. It is assumed in this paper that $\sigma_{\text{e}_\mathrm{ref}}^{2} = \sigma_{\text{e}_\mathrm{tra}}^{2} = \sigma^{2}$. 
In addition, the diagonal reflection and transmission matrices can be defined as $\boldsymbol{\Phi}_{\text{r}, \text{e}}  \triangleq \alpha_\text{r} \operatorname {diag} \left(\zeta e^{j \theta_{\text{r},1}}, \ldots, \zeta e^{j \theta_{\text{r},{N}_{\text{r}}}}\right)$ and $\boldsymbol{\Phi}_{\text{t}, \text{e}}  \triangleq \alpha_\text{t} \operatorname {diag} \left({\sqrt{1-{\zeta}^2}}e^{j \theta_{\text{t},1}}, \ldots, {\sqrt{1-{\zeta}^2}} e^{j \theta_{\text{t}, {N}_{\text{t}}}}\right)$, respectively. 
It is also assumed that $\alpha_\text{r} = \alpha_\text{t} = \alpha$.

\begin{proposition}\label{Proposition_8}
The secrecy rates of the STAR-RIS-supported system at the reflection and transmission zones, denoted by $R_{\mathrm{sec}}^{\mathrm{ref}}$ and $R_{\mathrm{sec}}^{\mathrm{tra}}$, are given by the following expressions
\begin{multline}\label{rate_secrecy_ref}
R_{\mathrm{sec}}^{\mathrm{ref}} = \\
\max\Bigg[0, \log_{2}\left(1 + \frac{p \left(\sqrt{\beta_{\mathrm{s d_\mathrm{ref}}}} + N_\mathrm{r} \alpha_\mathrm{r} \zeta \sqrt{\beta_{\mathrm{sr}} \beta_{\mathrm{r d_\mathrm{ref}}}}\right)^{2}}{\sigma^{2}}\right) \\
- \log_{2}\left(1 + \frac{p {\beta_{\mathrm{s e_\mathrm{ref}}}}}
{p N_\mathrm{r}^2 \alpha_\mathrm{r}^2 \zeta^2 \beta_{\mathrm{sr}} \beta_{\mathrm{r e_\mathrm{ref}}} + \sigma^{2}}\right) \Bigg],
\end{multline}
and
\begin{multline}\label{rate_secrecy_tra}
R_{\mathrm{sec}}^{\mathrm{tra}} = \\
\max\Bigg[0, \log_{2}\Bigg(1 + \frac{p \left(\sqrt{\beta_{\mathrm{s d_\mathrm{tra}}}} + N_\mathrm{t} \alpha_\mathrm{t} \sqrt{1-\zeta^2} \sqrt{\beta_{\mathrm{sr}} \beta_{\mathrm{r d_\mathrm{tra}}}}\right)^{2}}{\sigma^{2}}\Bigg) \\
- \log_{2}\left(1 + \frac{p {\beta_{\mathrm{s e_\mathrm{tra}}}}}
{p N_\mathrm{t}^2 \alpha_\mathrm{t}^2 (1-\zeta^2) \beta_{\mathrm{sr}} \beta_{\mathrm{r e_\mathrm{tra}}} + \sigma^{2}}\right) \Bigg].
\end{multline}
\end{proposition}

\begin{proof}
The secrecy rates of the STAR-RIS in \eqref{rate_secrecy_ref} and \eqref{rate_secrecy_tra} are generally found as $R_{\mathrm{sec}}^{\mathrm{ref}} = \max\left[0, R_{\mathrm{d}}^{\mathrm{ref}} - R_{\mathrm{e}}^{\mathrm{ref}}\right]$ and $R_{\mathrm{sec}}^{\mathrm{tra}} = \max\left[0, R_{\mathrm{d}}^{\mathrm{tra}} - R_{\mathrm{e}}^{\mathrm{tra}}\right]$, respectively. In this setup, $R_{\mathrm{sec}}^{\mathrm{ref}}$ and $R_{\mathrm{sec}}^{\mathrm{tra}}$ are expressed as in \eqref{rate_ref} and \eqref{rate_tra}, while the rates at Es in the reflection and transmission zones, denoted by $R_{\mathrm{e}}^{\mathrm{ref}}$ and $R_{\mathrm{e}}^{\mathrm{tra}}$, can be expressed as in \eqref{rate_secrecy_ref_top} and \eqref{rate_secrecy_tra_top}, respectively.
\begin{figure*}[t!]
\begin{equation} \label{rate_secrecy_ref_top}
R_{\mathrm{e}}^{\mathrm{ref}} = \min_{\theta_{1}, \ldots, \theta_{N_\mathrm{r}}} \log_{2}\left(1 + \frac{p \left|h_{\mathrm{s {e_{ref}}}} \right|^{2}}
{p \left|\mathbf{h}_{\mathrm{s {r_{ref}}}}^{\mathrm{T}} \boldsymbol{\Phi}_{\text{r}, \text{e}} \mathbf{h}_{\mathrm{r e_\mathrm{ref}}}\right|^{2} + \sigma^{2}}\right) = 
\log_{2}\left(1 + \frac{p \left|h_{\mathrm{s {e_{ref}}}} \right|^{2}}
{p \big( \alpha_\text{r} \zeta \sum_{n_\text{r}=1}^{N_\text{r}} e^{j \theta_{n_\text{r}}}\left[\mathbf{h}_{\mathrm{s {r_{ref}}}}\right]_{n_\text{r}}\left[\mathbf{h}_{\mathrm{r e_\mathrm{ref}}}\right]_{n_\text{r}} \big)^{2} + \sigma^{2} } \right).
\end{equation}
\end{figure*}
and
\begin{figure*}[t!]
\begin{equation} \label{rate_secrecy_tra_top}
R_{\mathrm{e}}^{\mathrm{tra}} = \min_{\theta_{1}, \ldots, \theta_{N_\mathrm{t}}} \log_{2}\left(1 + \frac{p \left|h_{\mathrm{s {e_{tra}}}} \right|^{2}}
{p \left|\mathbf{h}_{\mathrm{s {r_{tra}}}}^{\mathrm{T}} \boldsymbol{\Phi}_{\text{t}, \text{e}} \mathbf{h}_{\mathrm{r e_\mathrm{tra}}}\right|^{2} + \sigma^{2}}\right) =
\log_{2}\left(1 + \frac{p \left|h_{\mathrm{s {e_{tra}}}} \right|^{2}}
{p \big( \alpha_\text{t} {\sqrt{1-{\zeta}^2}} \sum_{n_\text{t}=1}^{N_\text{t}} e^{j \theta_{n_\text{t}}}\left[\mathbf{h}_{\mathrm{s {r_{tra}}}}\right]_{n_\text{t}}\left[\mathbf{h}_{\mathrm{r e_\mathrm{tra}}}\right]_{n_\text{t}} \big)^{2} + \sigma^{2} } \right).
\end{equation}
\hrule
\end{figure*}
These two expressions can be formulated in terms of the channel gains as
\begin{equation}
R_{\mathrm{e}}^{\mathrm{ref}} = \log_{2}\left(1 + \frac{p {\beta_{\mathrm{s e_\mathrm{ref}}}}}
{p N_\text{r}^2 \alpha_\text{r}^2 \zeta^2 \beta_{\mathrm{sr}} \beta_{\mathrm{r e_\mathrm{ref}}} + \sigma^{2}}\right),
\end{equation}
and
\begin{equation}
R_{\mathrm{e}}^{\mathrm{tra}} = \log_{2}\left(1 + \frac{p {\beta_{\mathrm{s e_\mathrm{tra}}}}}
{p N_\text{t}^2 \alpha_\text{t}^2 (1-\zeta^2) \beta_{\mathrm{sr}} \beta_{\mathrm{r e_\mathrm{tra}}} + \sigma^{2}}\right),
\end{equation}
respectively. For any given phase-shifting matrices, $\boldsymbol{\Phi}_{\text{r}, \text{e}}$ and $\boldsymbol{\Phi}_{\text{t}, \text{e}}$, $\mathbf{h}_{\mathrm{s {r_{ref}}}}^{\mathrm{T}} \boldsymbol{\Phi}_{\text{r}, \text{e}} \mathbf{h}_{\mathrm{r e_\mathrm{ref}}} = \alpha_\text{r} \zeta \sum_{n_\text{r}=1}^{N_\text{r}} e^{j \theta_{n_\text{r}}}\left[\mathbf{h}_{\mathrm{s {r_{ref}}}}\right]_{n_\text{r}}\left[\mathbf{h}_{\mathrm{r e_\mathrm{ref}}}\right]_{n_\text{r}}$ and $\mathbf{h}_{\mathrm{s {r_{tra}}}}^{\mathrm{T}} \boldsymbol{\Phi}_{\text{t}, \text{e}} \mathbf{h}_{\mathrm{r e_\mathrm{tra}}} = \alpha_\text{t} {\sqrt{1-{\zeta}^2}} \sum_{n_\text{t}=1}^{N_\text{t}}e^{j \theta_{n_\text{t}}}\left[\mathbf{h}_{\mathrm{s {r_{tra}}}}\right]_{n_\text{t}}\left[\mathbf{h}_{\mathrm{r e_\mathrm{tra}}}\right]_{n_\text{t}}$, where $\left[\mathbf{h}_{\mathrm{r e_\mathrm{ref}}}\right]_{n_\text{r}}$ and $\left[\mathbf{h}_{\mathrm{r e_\mathrm{tra}}}\right]_{n_\text{t}}$ are the $n_\text{r}$-th components and $n_\text{t}$-th components, respectively. In addition, the $n_\text{r}$-th and $n_\text{t}$-th phase-shifts at the STAR-RIS are chosen as $\theta_{n_\text{r}}=\arg \left(h_{\mathrm{s e_\mathrm{ref}}}\right)-\arg \left(\left[\mathbf{h}_{\mathrm{s {r_{ref}}}}\right]_{n_\text{r}}\left[\mathbf{h}_{\mathrm{r e_\mathrm{ref}}}\right]_{n_\text{r}}\right)$ and $\theta_{n_\text{t}}=\arg \left(h_{\mathrm{s e_\mathrm{tra}}}\right)-\arg \left(\left[\mathbf{h}_{\mathrm{s {r_{tra}}}}\right]_{n_\text{t}}\left[\mathbf{h}_{\mathrm{r e_\mathrm{tra}}}\right]_{n_\text{t}}\right)$ so that the optimal solution is achieved. For simplicity, the following notations, $\sum_{n_\text{r}=1}^{N_\text{r}}\left|\left[\mathbf{h}_{\mathrm{s {r_{ref}}}}\right]_{n_\text{r}}\left[\mathbf{h}_{\mathrm{r e_\mathrm{ref}}}\right]_{n_\text{r}}\right|={N_\text{r}}\sqrt{\beta_{\mathrm{sr}} \beta_{\mathrm{r e_\mathrm{ref}}}}$, $\sum_{n_\text{t}=1}^{N_\text{t}}\left|\left[\mathbf{h}_{\mathrm{s {r_{tra}}}}\right]_{n_\text{t}}\left[\mathbf{h}_{\mathrm{r e_\mathrm{tra}}}\right]_{n_\text{t}}\right|={N_\text{t}}\sqrt{\beta_{\mathrm{sr}} \beta_{\mathrm{r e_\mathrm{tra}}}}$, ${\left|h_{\mathrm{s e_\mathrm{ref}}}\right|^{2}} = \beta_{\mathrm{s e_\mathrm{ref}}}$ and ${\left|h_{\mathrm{s e_\mathrm{tra}}}\right|^{2}} = \beta_{\mathrm{s e_\mathrm{tra}}}$, are used. Therefore, the aforementioned secrecy rate expressions \eqref{rate_secrecy_ref} and \eqref{rate_secrecy_tra} can be readily derived.
\end{proof}

\section{Numerical Results and Discussion}
\label{sec4}
In this section, we compare the performance of STAR-RIS, HD-DF relaying, FD-DF relaying, and RIS-supported systems in terms of both rate and secrecy rate. The evaluation is conducted by varying several key parameters, including the reflection-to-transmission power ratio $\zeta$, the distance $d$, the number of elements $N_\mathrm{ref}$, the transmit power $p$, and the element-splitting factor $K$\footnote{The element-splitting factor represents the distribution ratio between the reflecting elements and the total number of elements in the STAR-RIS, with values in the range $0 < K < 1$.}, where $N_\text{r} = K {N_\mathrm{ref}}$ and $N_\text{t} = (1 - K) {N_\mathrm{ref}}$. For the sake of fair comparisons, the R is placed at the same location as the RIS and STAR-RIS units. Following \cite{REF_Emil_3} and \cite{9569598}, both line-of-sight (LoS) and non-LOS (NLoS) versions of the 3GPP Urban Micro are adopted to characterize the channel gains for distances of at least ten meters. Since shadow fading is neglected, the channel gains are expressed as a function of the carrier frequency $f_c$ (in GHz) and the distance $d$ (in meters), given by
\begin{equation}
\begin{aligned}
&\beta(d) \, [\mathrm{dB}] = G_{\text{T}} + G_{\text{R}} + \\
& \begin{cases}-28-20\log_{10}(f_c) - 22 \log_{10}(d) & \text {if LoS}, \\
-22.7 - 26\log_{10}(f_c) - 36.7 \log_{10}(d) & \text {if NLoS},
\end{cases}
\end{aligned}
\end{equation}
where $G_{\text{T}}$ and $G_{\text{R}}$ denote the antenna gains (in dBi) at the transmitter and receiver, respectively. Equal-sized 5 $\mathrm{dBi}$ antennas at S, R, RIS and STAR-RIS are assumed, while omnidirectional antennas with 0 $\mathrm{dBi}$ are assumed at the legitimate users and Es. 
As shown in Figs. \ref{figure_1} and \ref{figure_2}, the legitimate Ds in the reflection and transmission zones are denoted by $\text{U}_{\text{r}}$ and $\text{U}_{\text{t}}$, respectively, while the Es are denoted in Fig. \ref{figure_2} by $\text{E}_{\text{r}}$ and $\text{E}_{\text{t}}$. Additionally, the distances between S and Ds in the reflection and transmission zones are defined by the variables $d_\text{r}$ and $d_\text{t}$, respectively, and the distances between S and Es in the reflection and transmission zones are defined by $d_\text{r,e}$ and $d_\text{t,e}$, respectively, while R, RIS and STAR-RIS are placed at fixed locations, and S is located at the origin. In addition, $d_\text{sr}$ denotes the distance between S and R, while $d_{\text{s} \text{d}_\text{DF}}$ and $d_{\text{s} \text{e}_\text{DF}}$ for DF relaying cases and both $d_\text{r}$ and $d_\text{t}$ for STAR-RIS setup can be found using the Pythagoras theorem to be used in calculating the channel gains, given the values of $d_{\text{s} \text{d}_\text{r}}$, $d_{\text{s} \text{d}_\text{t}}$, $d_{\text{s} \text{e}_\text{r}}$, $d_{\text{s} \text{e}_\text{t}}$, $d_\textit{v}$ and $\it{d}_{v_e}$. 
The noise power is considered as $-174 + 10 \log_{10}(B) + \textit{NF}$ in dBm, where it is assumed that the bandwidth is $B=10$ MHz, the noise figure is $\textit{NF} = 10$ dB, and the carrier frequency $f_c = 3$ GHz. We assume $\alpha_\text{r} = \alpha_\text{t} = 1$, and a residual loop interference channel gain of $\beta_{\text{li}} = -130$ dB, representing nearly perfect self-interference cancellation.

\begin{figure*}[t]
\centering
\subfloat[\label{figure_3a}]{\includegraphics[width=8.7cm,height=6.5cm]{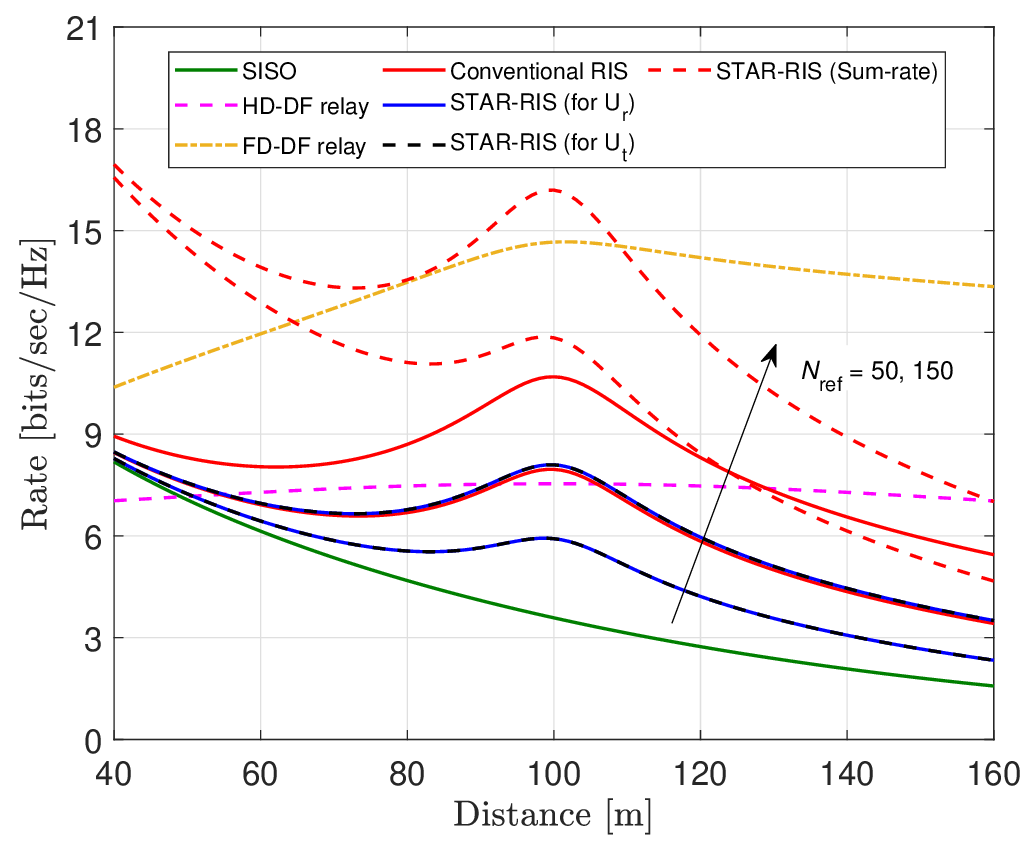}}
\subfloat[\label{figure_3b}]
{\includegraphics[width=8.7cm,height=6.5cm]{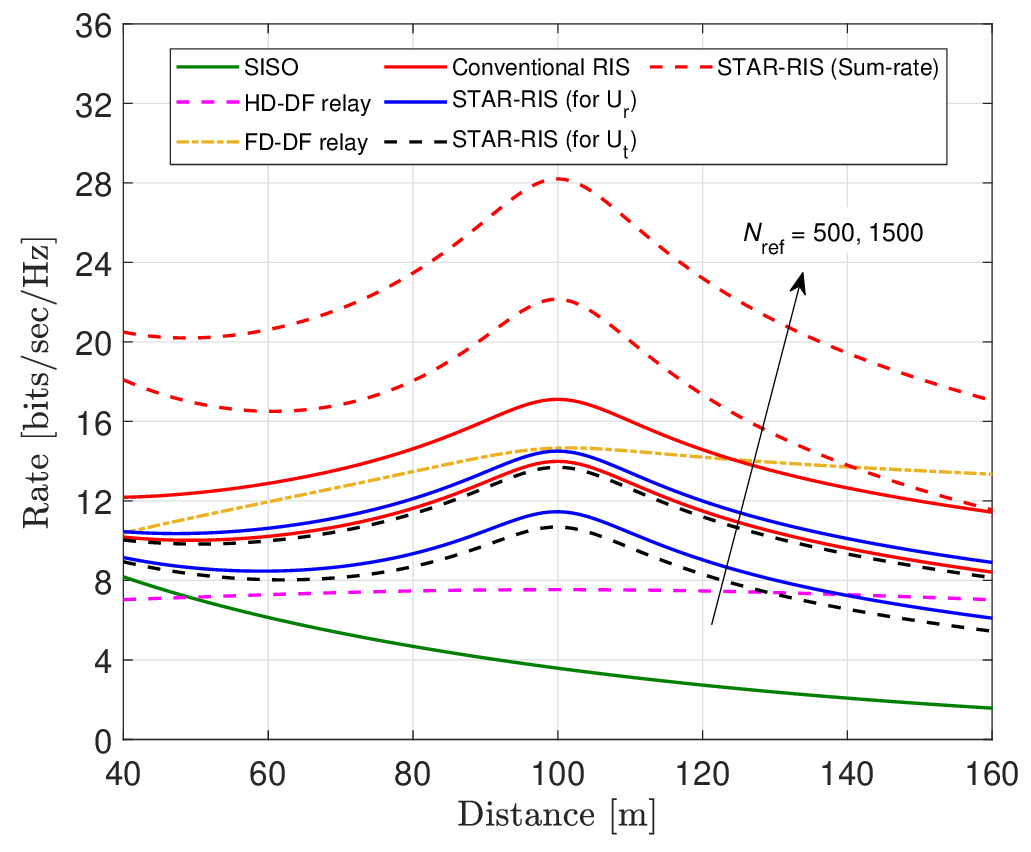}}
\caption{Achievable rate performance versus the distances, $d_{\text{s} \text{d}_\text{r}} = d_{\text{s} \text{d}_\text{t}}$, with $p$ = 20 dBm, $\it{d}_\text{sr}$ = 100 m, $\it{d}_v$ = 10 m, and $K$ = 0.5: (a) $\zeta = 1/{\sqrt{2}}$ and $N_\mathrm{ref} = [50, 150]$,(b) $\zeta = 0.8$ and $N_\mathrm{ref} = [500, 1500]$.}
\label{figure_3}
\end{figure*}

Fig. \ref{figure_3} illustrates the rate performance comparison among STAR-RIS and the SISO, RIS, HD-DF relaying and FD-DF relaying systems versus the distances $d_{\text{s} \text{d}_\text{r}}$ and $d_{\text{s} \text{d}_\text{t}}$, with different values of $\zeta$ and ${N_\mathrm{ref}}$. As observed, the SISO systems yield the lowest rate, while the FD-DF relaying system achieves the highest rate when the number of elements is small. Moreover, increasing the number of elements improves the rate performance of both STAR-RIS and RIS. As found from \eqref{Nr_max}-\eqref{Nt_max} and \cite[Eq.~(15)]{REF_Emil_3}, and illustrated in Fig. \ref{figure_3a} for $\zeta = 1/{\sqrt{2}}$, it can be noticed that at least 58, 58 and 41 elements are needed for the STAR-RIS-$\text{U}_{\text{r}}$, STAR-RIS-$\text{U}_{\text{t}}$ and RIS, respectively, to outperform HD-DF relaying when $d_{\text{s} \text{d}_\text{r}} = d_{\text{s} \text{d}_\text{t}} = d_{\text{sr}}$. In this case, the same number of elements is required to achieve the same rate at $\text{U}_{\text{t}}$ and $\text{U}_{\text{r}}$ when $\zeta = 1/{\sqrt{2}}$, (i.e., $\zeta = \sqrt{1 - \zeta^2}$). On the other hand, as found from \eqref{eq19}-\eqref{eq21} and shown in Fig. \ref{figure_3b}, it can be observed that at least 792, 1056, and 633 elements are required, respectively, for STAR-RIS-$\text{U}_{\text{r}}$, STAR-RIS-$\text{U}_{\text{t}}$, and conventional RIS, to provide a higher achievable rate than the FD-DF relaying scheme. In this case, the STAR-RIS-$\text{U}_{\text{r}}$ consistently achieves a higher rate than STAR-RIS-$\text{U}_{\text{t}}$. This behavior arises because $\zeta > 1/\sqrt{2}$ in Fig. \ref{figure_3b}, resulting in the amplitude coefficient of $\boldsymbol{\Phi}_{\text{r}, \text{d}}$ being larger than that of $\boldsymbol{\Phi}_{\text{t}, \text{d}}$. It can also be observed that the STAR-RIS system offers a clear advantage over the conventional RIS in terms of sum rate. 

Fig. \ref{figure_4} depicts the achievable rates of the STAR-RIS and different benchmarks as functions of the transmit power. As expected, the achievable rates improve monotonically with increasing the transmit power, which is intuitive. At low transmit power levels, the FD-DF relaying scheme slightly outperforms the other systems. However, its rate improvements appear to diminish at higher transmit power levels, indicating a transition from a power-limited into an interference-limited regime, where its rate performance saturates and eventually becomes the worst among the considered schemes. Overall, these findings highlight the scalability advantage of both STAR-RIS- and RIS-assisted systems, which exhibit superior performance at higher transmit power budgets and larger element deployments.

Fig. \ref{figure_5} demonstrates the achievable rate performance of the STAR-RIS and benchmarks versus $\zeta$ for different values of $N_\mathrm{ref}$. It illustrates a non-monotonic relationship between the achievable rate of STAR-RIS and $\zeta$, while the rates of SISO, RIS, and HD- and FD-DF relaying systems remain independent of $\zeta$. Therefore, this demonstrates the existence of a fundamental trade-off between the rate and $\zeta$, with the optimal value $\zeta = 1/{\sqrt{2}}$ balancing this trade-off. Notably, the optimal value of $\zeta = 1/{\sqrt{2}}$ remains constant across different values of $N_\mathrm{ref}$. This figure also shows that both STAR-RIS and RIS systems outperform the HD-DF relaying scheme with the lowest value of $N_\mathrm{ref}$, and indicates high sensitivity to the value of $\zeta$. In contrast, the results with the highest value of $N_\mathrm{ref}$ show better rate performance, where STAR-RIS and RIS systems outperform the FD-DF relaying scheme. Therefore, increasing the number of elements yields significant performance gains and enhances the robustness of STAR-RIS-and RIS-assisted systems.

The achievable rate performance of the STAR-RIS, SISO, RIS, HD-DF relaying, and FD-DF relaying systems versus the number of elements is shown in Fig. \ref{figure_6} for two different values of the transmit power. It is evident that the transmit power impacts the rate performance of all systems, and that increasing both the number of elements and the transmit power enables STAR-RIS- and RIS-assisted systems to outperform the FD-DF relaying scheme. These results confirm that fully exploiting the benefits of STAR-RIS and RIS technologies requires proper adjustment of the number of elements and careful selection of the transmit power.

In Fig. \ref{figure_7}, the achievable rate performance of the STAR-RIS, SISO, RIS, HD-DF relaying, and FD-DF relaying systems is illustrated as a function of the element-splitting factor $K$, for different values of $N_\mathrm{ref}$. It is clear from this figure that both conventional RIS and STAR-RIS systems show a strong dependence on $K$. By contrast, the SISO, HD- and FD-DF relaying systems exhibit flat curves, since their rate performance is independent of $K$. Two STAR-RIS configurations (i.e., for $\text{U}_\text{r}$ and $\text{U}_\text{t}$) are considered in this comparison, and both achieve superior rate performance compared to the HD- and FD-DF relaying systems, particularly when $N_\mathrm{ref}$ is large\footnote{Note that the performance of $\text{U}_\text{r}$ is the complement of $\text{U}_\text{t}$ due to $N_\text{t} = (1 - K) N_\mathrm{ref}$.}. It is also observed that increasing $K$ initially enhances the achievable rate of the STAR-RIS system ($\text{U}_\text{r}$), while the performance gradually saturates with increasing $K$. In addition, with a small number of elements, both STAR-RIS and RIS systems may provide only modest gains over the HD- and FD-DF relaying schemes. However, as the number of elements increases, the rate of STAR-RIS and RIS improves substantially, eventually outperforming the FD-DF relaying scheme, which suffers from interference limitations at high rate values. Notably, the optimal value of $K$, at which the rates of STAR-RIS-$\text{U}_{\text{t}}$ and STAR-RIS-$\text{U}_{\text{r}}$ coincide, remains unchanged across different numbers of elements; however, the achievable rate at this point increases significantly with larger $N_\mathrm{ref}$. Overall, Fig. \ref{figure_7} shows that the STAR-RIS system is highly sensitive to the choice of $K$ and benefits significantly from the larger number of elements. This highlights its scalability advantage over classical relaying schemes.

Fig. \ref{figure_8} shows the secrecy rate performance of the STAR-RIS and benchmarks against the transmit power $p$, under different values of $N_\mathrm{ref}$ and $\zeta$. In particular, in Fig. \ref{figure_8a}, with $N_\mathrm{ref}=100$ elements and $\zeta=1/{\sqrt{2}}$, the STAR-RIS and RIS systems offer only limited secrecy rate improvements, whereas the FD-DF relaying scheme outperforms all other systems in the low and moderate transmit power regimes. However, at higher transmit power, both STAR-RIS and RIS systems exhibit a steeper secrecy rate improvement. Nevertheless, their performance remains ultimately limited by the finite number of passive elements. In contrast, Fig. \ref{figure_8b} with $N_\mathrm{ref}=1000$ elements and $\zeta=0.5$ shows that increasing $N_\mathrm{ref}$ to 1000 enables the STAR-RIS and RIS systems to achieve substantial secrecy rate gains, surpassing the FD-DF relaying scheme at medium and high transmit power levels. In addition, the STAR-RIS-$\text{U}_\text{r}$ and STAR-RIS-$\text{U}_\text{t}$ demonstrate superior scalability and robustness, highlighting the critical role of deploying a large number of elements in enhancing physical-layer security. Moreover, the STAR-RIS results in Fig. \ref{figure_8} highlight the impact of $\zeta$ on secrecy rate performance. For $\zeta = 1/\sqrt{2}$, both STAR-RIS-$\text{U}_\text{r}$ and STAR-RIS-$\text{U}_\text{t}$ yield identical secrecy rates since $\zeta = \sqrt{1-\zeta^{2}}$. In contrast, for $\zeta=0.5$, the STAR-RIS-$\text{U}_\text{t}$ achieves a higher secrecy rate than STAR-RIS-$\text{U}_\text{r}$. This is because $\zeta < 1/\sqrt{2}$ in Fig. \ref{figure_8b}, resulting in the amplitude coefficient of $\boldsymbol{\Phi}_{\text{r}, \text{d}}$ being smaller than that of $\boldsymbol{\Phi}_{\text{t}, \text{d}}$. Overall, these findings highlight two key observations: (i) the FD-DF relaying scheme is advantageous in the low-power regime but its rate performance levels off due to residual loop interference; and (ii) the secrecy rate performance of both STAR-RIS and RIS systems strongly depends on the number of elements. Hence, scaling the number of elements enhances the beamforming resolution required to achieve significant secrecy rate gains at higher transmit power levels. It can also be noticed from this figure that the STAR-RIS system gradually exhibits better performance than the conventional RIS in terms of secrecy sum rate, especially at higher transmit power levels.

\begin{figure}[t]
\centering
{\includegraphics[width=8.5cm,height=6.5cm]{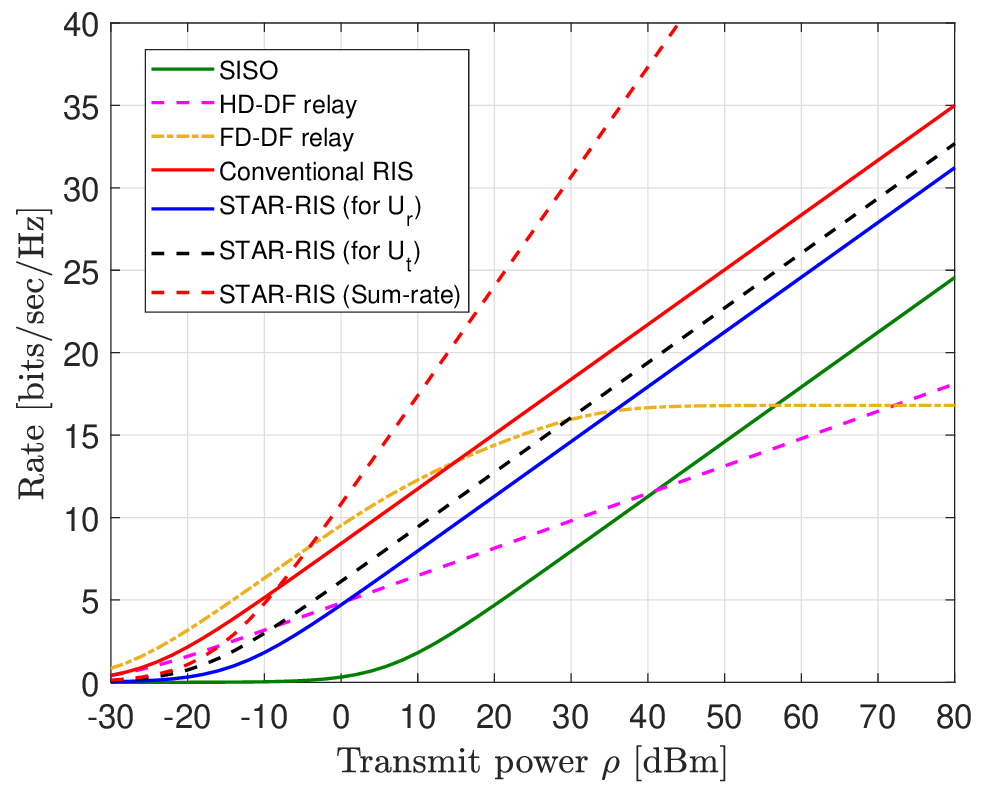}}
\caption{Achievable rate performance versus the transmit power, $p$, with $d_{\text{s} \text{d}_\text{r}} = d_{\text{s} \text{d}_\text{t}}$ = 80 m, $\it{d}_\text{sr}$ = 60 m, $\it{d}_v$ = 10 m, $N_\mathrm{ref}$ = 1000 elements, $K$ = 0.5, and $\zeta$ = 0.5.}
\label{figure_4}
\end{figure}

\begin{figure}[t!]
\centering
{\includegraphics[width=8.5cm,height=6.5cm]{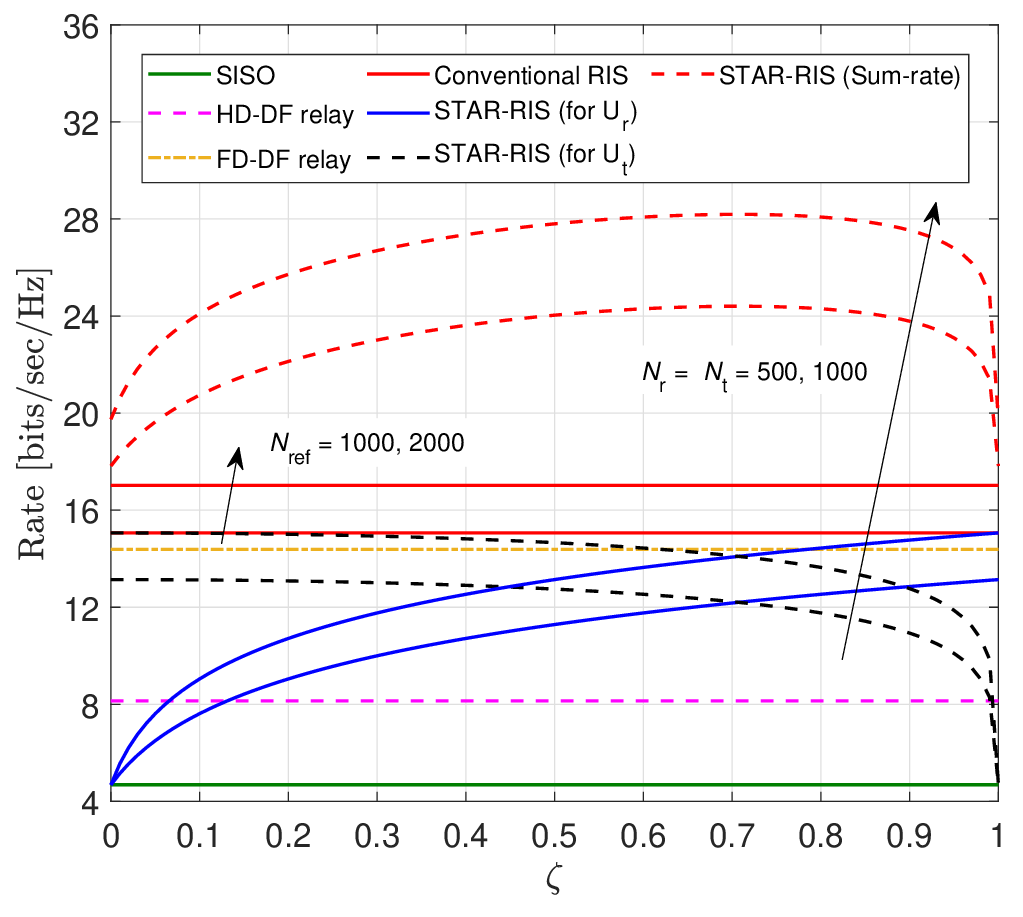}}
\caption{Achievable rate performance versus $\zeta$, with $d_{\text{s} \text{d}_\text{r}} = d_{\text{s} \text{d}_\text{t}}$ = 80 m, $\it{d}_\text{sr}$ = 60 m, $\it{d}_v$ = 10 m, $K$ = 0.5, $p$ = 20 dBm, and different values of $N_\mathrm{ref}$.}
\label{figure_5}
\end{figure}

\begin{figure}[t]
\centering
\subfloat[\label{figure_6a}]{\includegraphics[width=4.5cm,height=5.8cm]{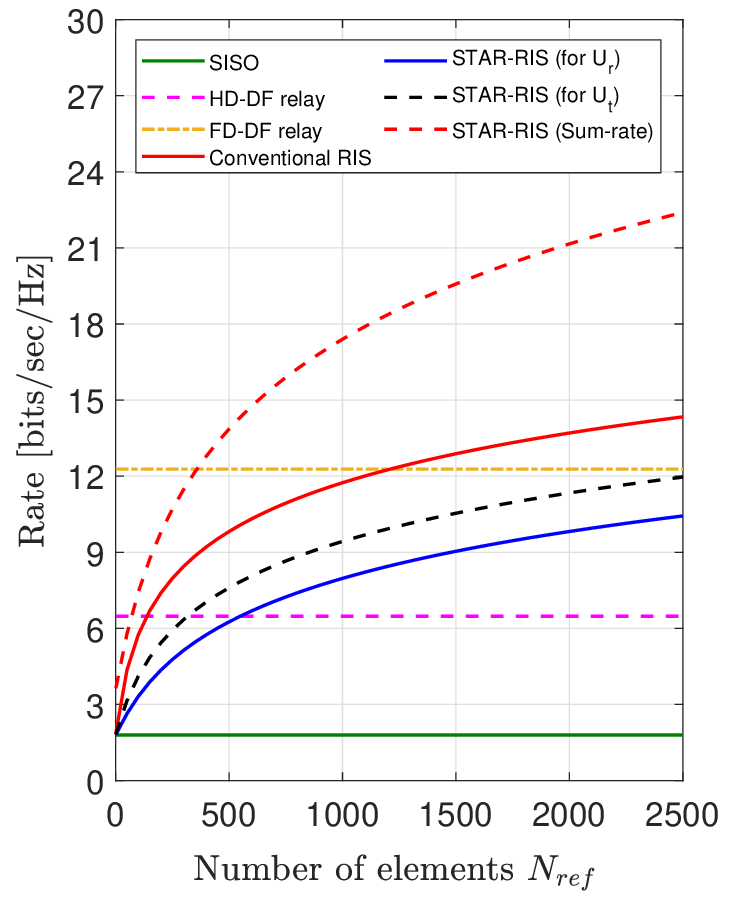}}
\subfloat[\label{figure_6b}]
{\includegraphics[width=4.5cm,height=5.8cm]{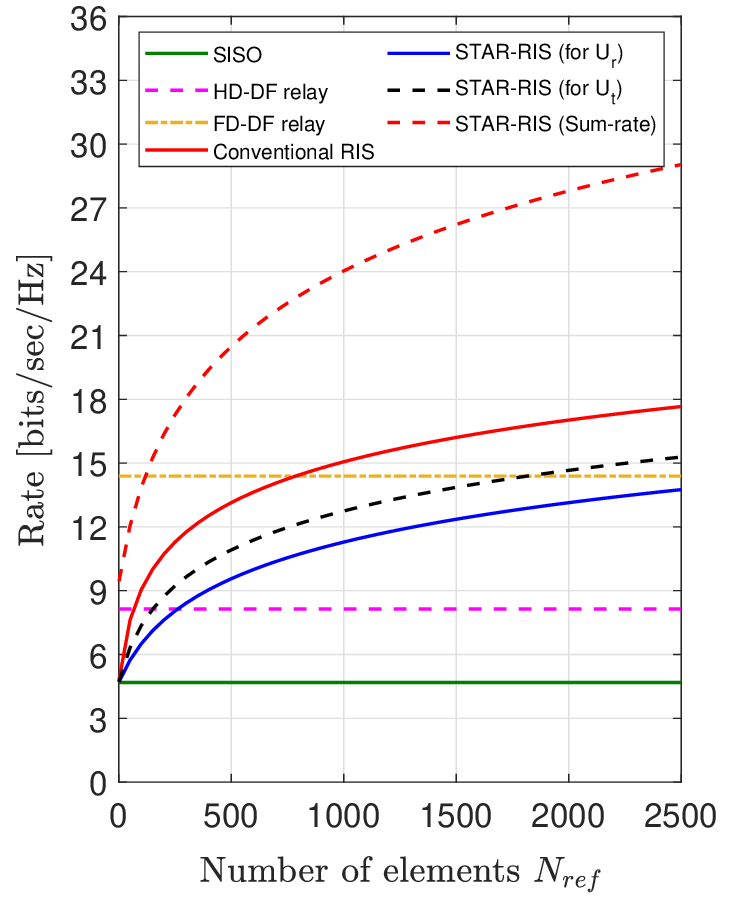}}
\caption{Achievable rate performance versus the number of elements, $N_\mathrm{ref}$, with $d_{\text{s} \text{d}_\text{r}} = d_{\text{s} \text{d}_\text{t}}$ = 80 m, $\it{d}_\text{sr}$ = 60 m, $\it{d}_v$ = 10 m, $K$ = 0.5 and $\zeta$ = 0.5: (a) $p$ = 10 dBm, (b) $p$ = 20 dBm.}
\label{figure_6}
\end{figure}

\begin{figure}[t!]
\centering
{\includegraphics[width=8.5cm,height=6.5cm]{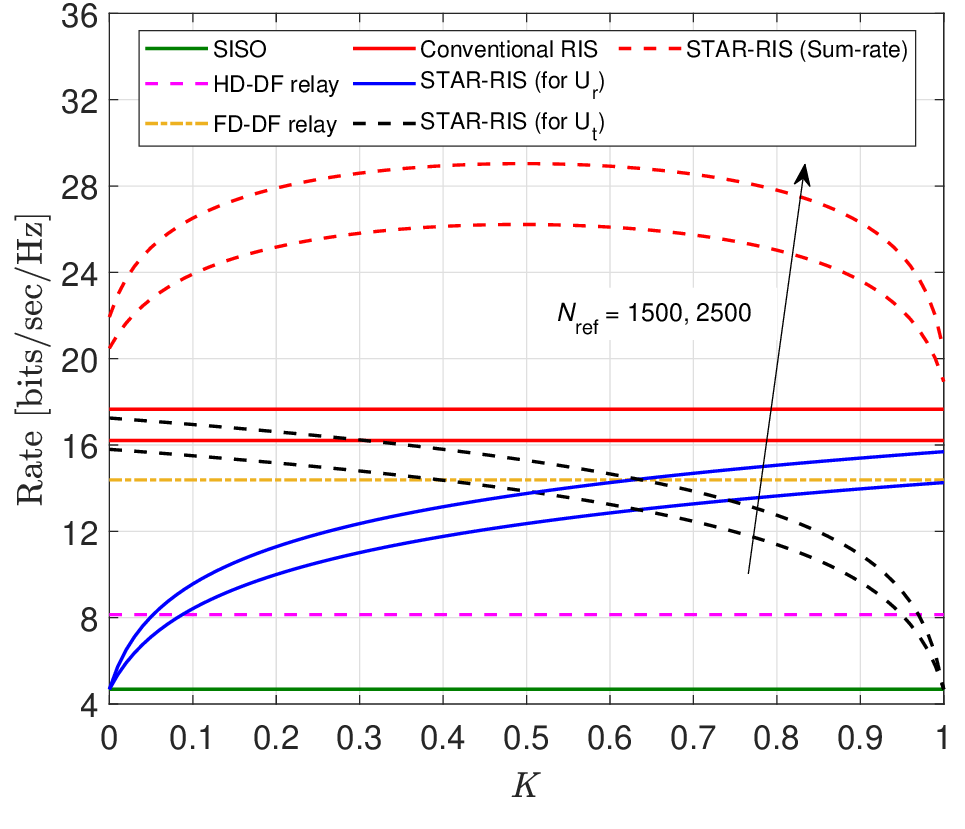}}
\caption{Achievable rate performance versus $K$, with $d_{\text{s} \text{d}_\text{r}} = d_{\text{s} \text{d}_\text{t}}$ = 80 m, $\it{d}_\text{sr}$ = 60 m, $\it{d}_v$ = 10 m, $p$ = 20 dBm, $\zeta$ = 0.5 and different values of $N_\mathrm{ref}$.}
\label{figure_7}
\end{figure}

\begin{figure}[t]
\centering
\subfloat[\label{figure_8a}]{\includegraphics[width=8.5cm,height=6.5cm]{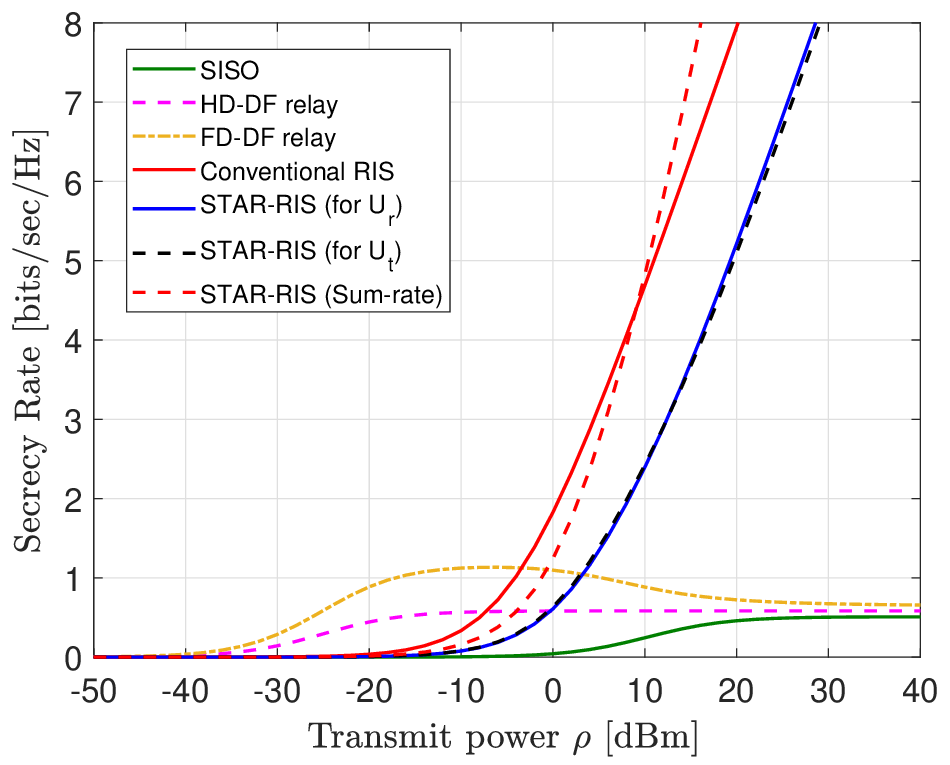}}
\hspace{\fill}
\subfloat[\label{figure_8b}]
{\includegraphics[width=8.5cm,height=6.5cm]{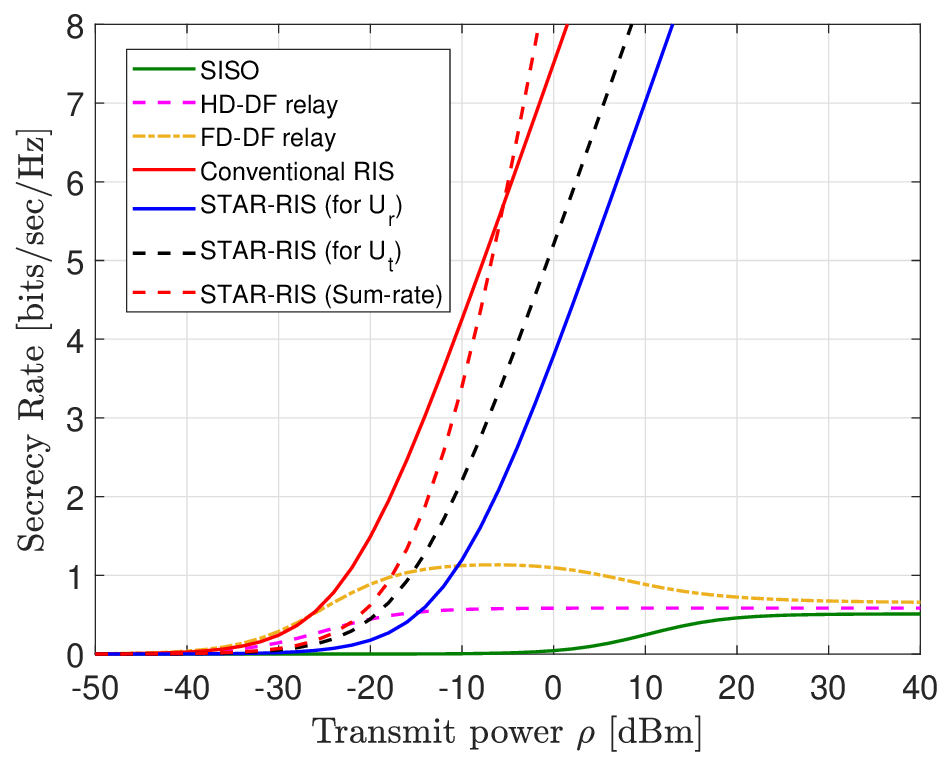}}
\caption{Secrecy rate performance versus the transmit power, $p$, with $d_{\text{s} \text{d}_\text{r}}$ = $d_{\text{s} \text{d}_\text{t}}$ = 100 m, $d_{\text{s} \text{e}_\text{r}}$ = 110 m,
$d_{\text{s} \text{e}_\text{t}}$ = 120 m, $\it{d}_\text{sr}$ = 80 m, $\it{d}_v$ = 10 m, $\it{d}_{v_e}$ = 12 m, and $K$ = 0.5: (a) $N_\mathrm{ref}$ = 100 elements and $\zeta = 1/{\sqrt{2}}$, (b) $N_\mathrm{ref}$ = 1000 elements and $\zeta$ = 0.5.}
\label{figure_8}
\end{figure}

\begin{figure}[t]
\centering
\subfloat[\label{figure_9a}]{\includegraphics[width=8.5cm,height=6.5cm]{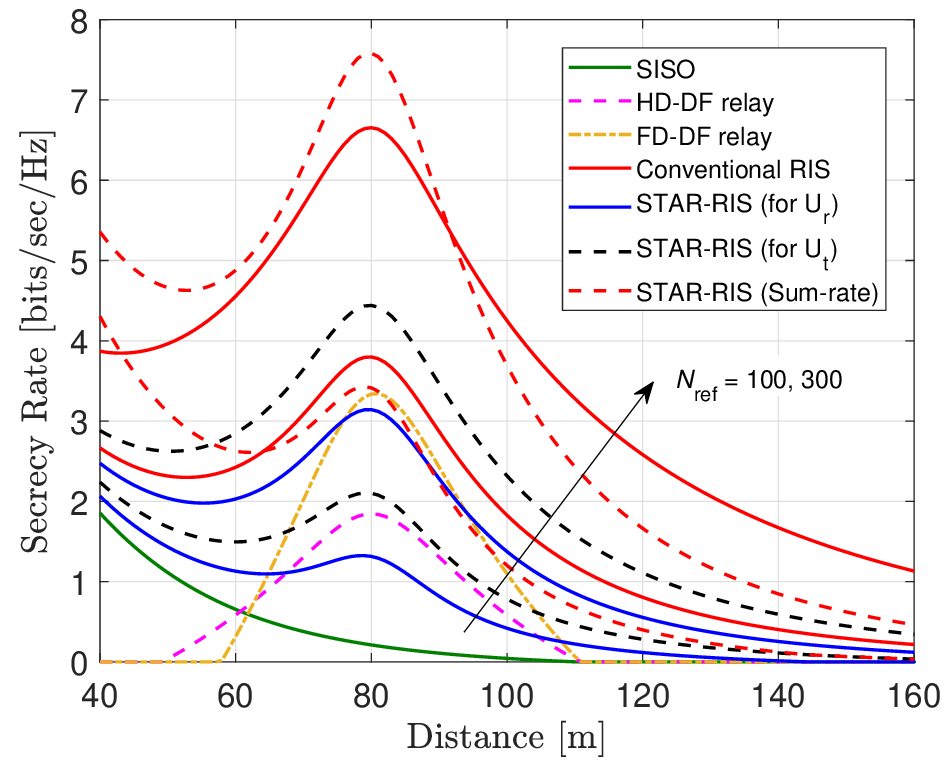}}
\hspace{\fill}
\subfloat[\label{figure_9b}]
{\includegraphics[width=8.5cm,height=6.5cm]{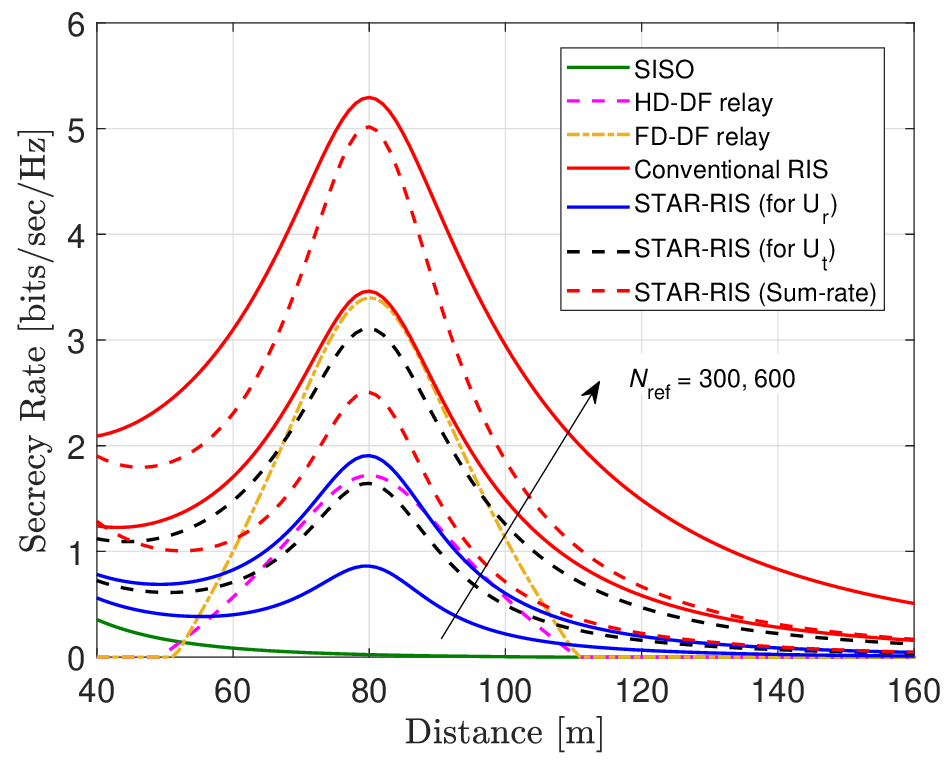}}
\caption{Secrecy rate performance versus the distances, $d_{\text{s} \text{d}_\text{r}} = d_{\text{s} \text{d}_\text{t}}$, with $d_{\text{s} \text{e}_\text{r}}$ = 110 m,
$d_{\text{s} \text{e}_\text{t}}$ = 120 m, $\it{d}_\text{sr}$ = 80 m, $\it{d}_v$ = 10 m, $\it{d}_{v_e}$ = 12 m, $\zeta$ = 0.5 and $K$ = 0.5: (a) with $p$ = 0 dBm and $N_\mathrm{ref} = [100, 300]$ elements, (b) with $p$ = -10 dBm and $N_\mathrm{ref} = [300, 600]$ elements.}
\label{figure_9}
\end{figure}

In Fig. \ref{figure_9}, the secrecy rate performance of the STAR-RIS, SISO, RIS, HD-DF relaying, and FD-DF relaying systems is shown versus the distances $d_{\text{s} \text{d}_\text{r}} = d_{\text{s} \text{d}_\text{t}}$ under different values of $p$ and $N_\mathrm{ref}$. It can be clearly seen from both subfigures \ref{figure_9a} and \ref{figure_9b} that the SISO system achieves a very limited secrecy rate, and the highest secrecy rate is at $d_{\text{s} \text{d}_\text{r}} = d_{\text{s} \text{d}_\text{t}} = d_\text{sr}$. It is also shown that the FD-DF relaying scheme attains the highest secrecy rate in the mid-range distances (around 80–100 m), where R is best positioned to balance S-R and R-D links. However, its secrecy rate performance drops significantly outside this range, particularly as the D moves farther from S. Additionally, the STAR-RIS and RIS systems achieve lower peak secrecy rates than FD-DF relaying scheme with a small number of elements, but they provide more stable performance over a wider distance range. Importantly, increasing the number of elements enhances the secrecy rate of STAR-RIS and RIS systems. 
This figure also indicates that under a lower transmit power, the STAR-RIS and RIS systems require a larger number of elements to surpass the performance of HD- and FD-DF relaying schemes.

\begin{figure}[t]
\centering
{\includegraphics[width=8.5cm,height=6.3cm]{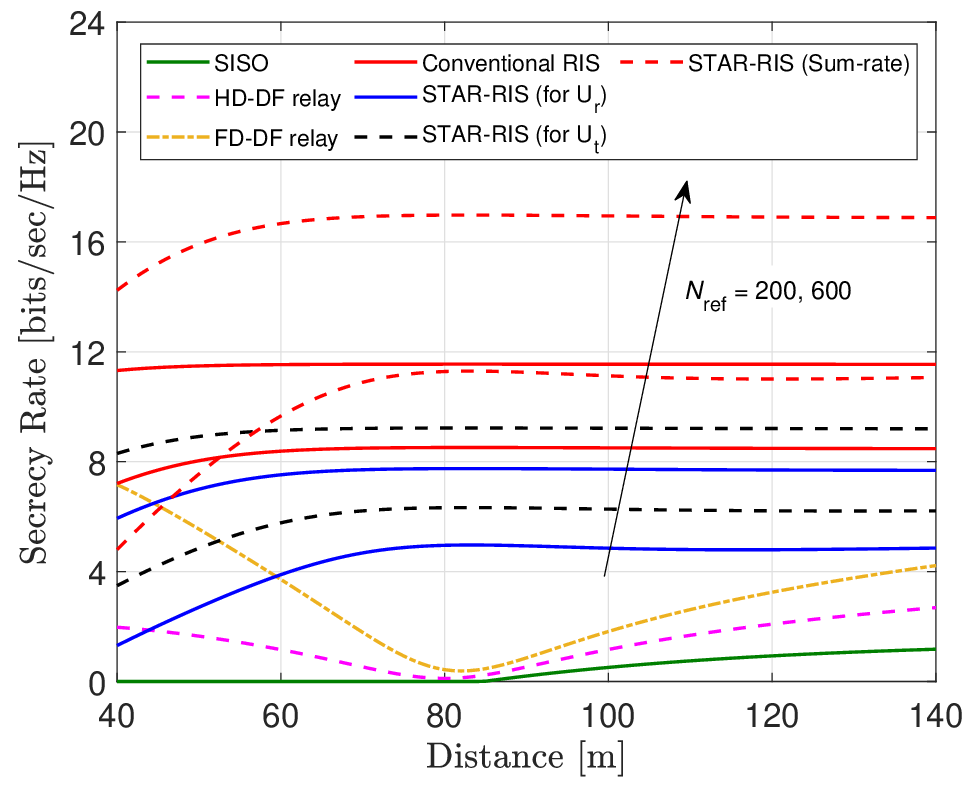}}
\caption{Secrecy rate performance versus the distances, $d_{\text{s} \text{e}_\text{r}} = d_{\text{s} \text{e}_\text{t}}$, with $d_{\text{s} \text{d}_\text{r}}$ = $d_{\text{s} \text{d}_\text{t}}$ = 85 m, $\it{d}_\text{sr}$ = 80 m, $\it{d}_v$ = 10 m, $\it{d}_{v_e}$ = 12 m, $\zeta$ = 0.5, $K$ = 0.5, $p$ = 10 dBm and varying $N_\mathrm{ref}$.}
\label{figure_10}
\end{figure}

\begin{figure}[t!]
\centering
{\includegraphics[width=8.5cm,height=6.3cm]{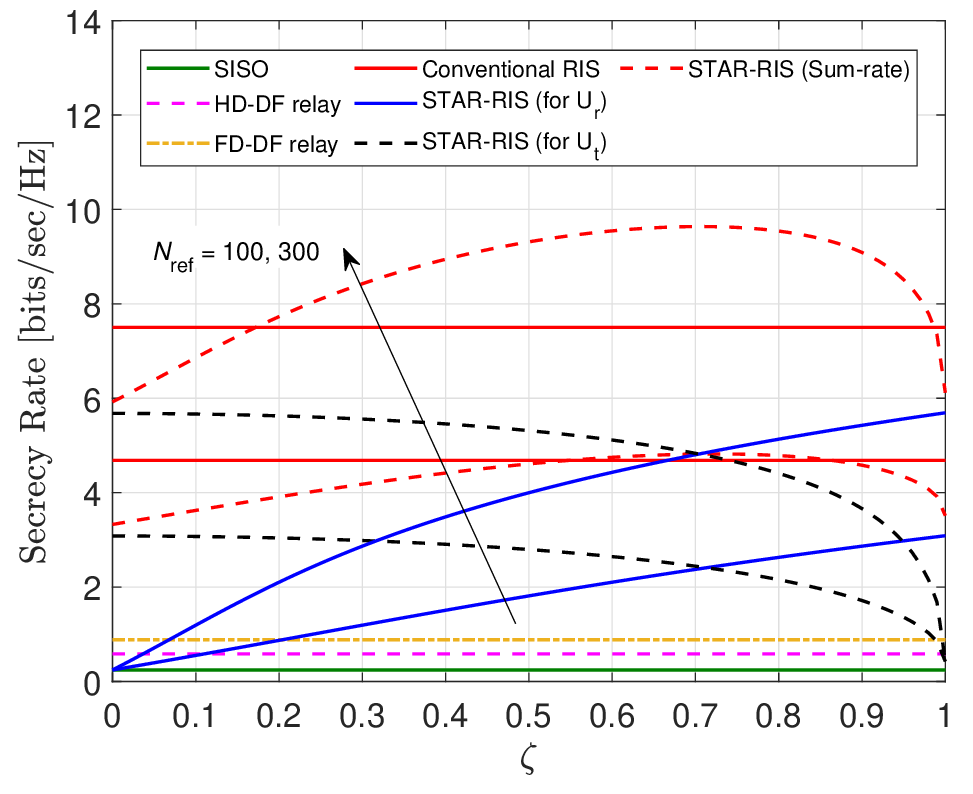}}
\caption{Secrecy rate performance versus $\zeta$, with $d_{\text{s} \text{d}_\text{r}}$ = $d_{\text{s} \text{d}_\text{t}}$ = 100 m, $d_{\text{s} \text{e}_\text{r}}$ = 110 m,
$d_{\text{s} \text{e}_\text{t}}$ = 120 m, $\it{d}_\text{sr}$ = 80 m, $\it{d}_v$ = 10 m, $\it{d}_{v_e}$ = 12 m, $K$ = 0.5, $p$ = 10 dBm, and different values of $N_\mathrm{ref}$.}
\label{figure_11}
\end{figure}

Fig. \ref{figure_10} shows the secrecy rate performance of the STAR-RIS, SISO, RIS, HD-DF relaying, and FD-DF relaying systems versus the distances $d_{\text{s} \text{e}_\text{r}} = d_{\text{s} \text{e}_\text{t}}$ with a different number of elements. This figure reveals that the secrecy rate remains low for all configurations when $d_{\text{s} \text{e}} < d_{\text{s} \text{r}}$, with the lowest secrecy rate values observed in SISO, HD- and FD-DF relaying schemes at $d_{\text{s} \text{e}} = d_{\text{s} \text{r}}$. It is also clear that the FD-DF relaying scheme outperforms the STAR-RIS system at low values of $d_{\text{s} \text{e}}$. However, in SISO, HD- and FD-DF relaying schemes, the secrecy rate improves slightly as Es move further away from the R. In contrast, deploying a large number of elements in STAR-RIS and RIS systems yields substantial performance gains. Moreover, for any given distance, the secrecy rate is consistently higher for STAR-RIS and RIS systems with a larger number of elements. This is mainly because larger STAR-RISs or RISs offer higher beamforming gains and enhanced flexibility in shaping the wireless propagation environment. Therefore, this comparison shows that STAR-RIS and RIS systems are highly effective technologies for improving the physical layer security, with their performance enhancing significantly with increasing the number of elements.

The secrecy rate performance of the STAR-RIS and benchmarks versus $\zeta$ is depicted in Fig. \ref{figure_11} with varying the number of elements. A similar trend is seen in Fig. \ref{figure_5}, in which one can confirm that there is a non-monotonic relationship between the secrecy rate and $\zeta$, while the secrecy rate performance of the SISO, RIS, HD- and FD-DF relaying systems remains independent of $\zeta$. It can be seen from this figure that the optimal value $\zeta = 1/{\sqrt{2}}$ balances the trade-off between the secrecy rate and $\zeta$, and remains constant across different values of $N_\mathrm{ref}$. It is also noticeable that both STAR-RIS and RIS systems outperform the HD-DF relaying scheme with low values of $N_\mathrm{ref}$, and the secrecy rate performance improves and outperforms the FD-DF relaying scheme as the number of elements increases. Consistent with the trends observed in Fig. \ref{figure_5}, which showed robustness against increasing the number of elements, these results confirm that increasing the number of elements in STAR-RIS and RIS systems is a fundamental strategy for achieving substantial performance gains in the secrecy rate and enhancing the resilience of RIS-assisted secure communication systems.

In Fig. \ref{figure_12}, the secrecy rate performance of the STAR-RIS and benchmarks versus the number of elements is presented for two different values of the transmit power. This figure shows that the secrecy rate increases with the transmit power across all systems, as observed in Fig. \ref{figure_6}. Furthermore, increasing the number of elements enables the STAR-RIS and RIS systems to outperform the HD- and FD-DF relaying schemes. Notably, unlike the results observed in Fig. \ref{figure_6}, a small number of elements and low transmit power are required for STAR-RIS and RIS systems to surpass HD- and FD-DF relaying systems. Therefore, these results confirm that carefully selecting the transmit power and properly adjusting the number of elements are crucial to increasing the secrecy rate gains of STAR-RIS and RIS systems.

\begin{figure}[t]
\centering
\subfloat[\label{figure_12a}]{\includegraphics[width=4.3cm,height=5.9cm]{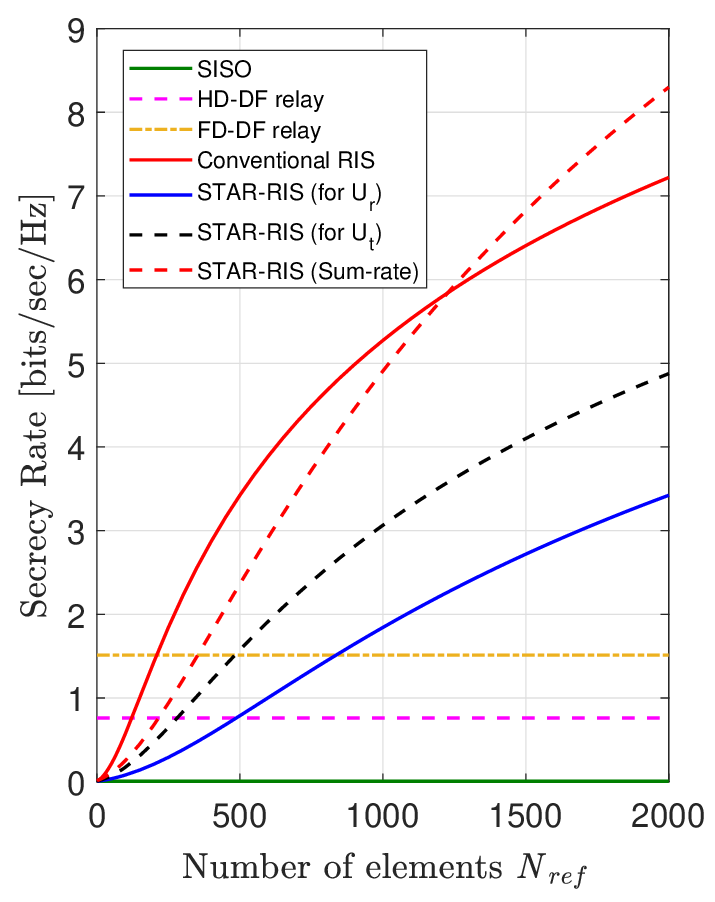}}
\subfloat[\label{figure_12b}]
{\includegraphics[width=4.3cm,height=5.9cm]{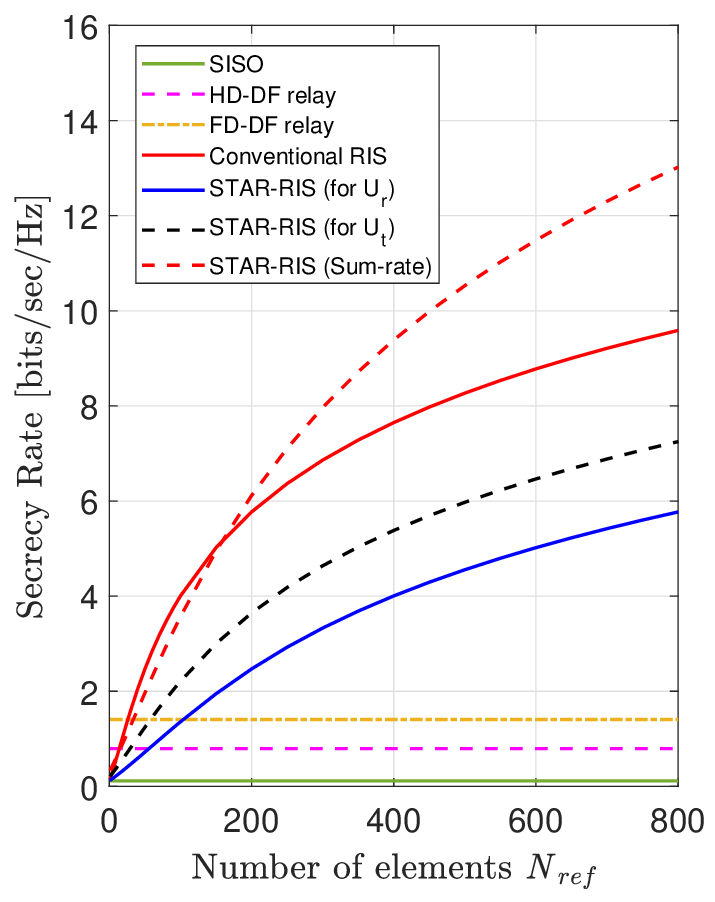}}
\caption{Secrecy rate performance versus the number of elements, $N_\mathrm{ref}$, with $d_{\text{s} \text{d}_\text{r}}$ = $d_{\text{s} \text{d}_\text{t}}$ = 100 m, $d_{\text{s} \text{e}_\text{r}}$ = 110 m,
$d_{\text{s} \text{e}_\text{t}}$ = 120 m, $\it{d}_\text{sr}$ = 90 m, $\it{d}_v$ = 10 m, $\it{d}_{v_e}$ = 12 m, $K$ = 0.5 and $\zeta$ = 0.5, (a) with $p$ = -10 dBm, (b) with $p$ = 5 dBm.}
\label{figure_12}
\end{figure}

\begin{figure}[t!]
\center
{\includegraphics[width=8.5cm,height=6.3cm]{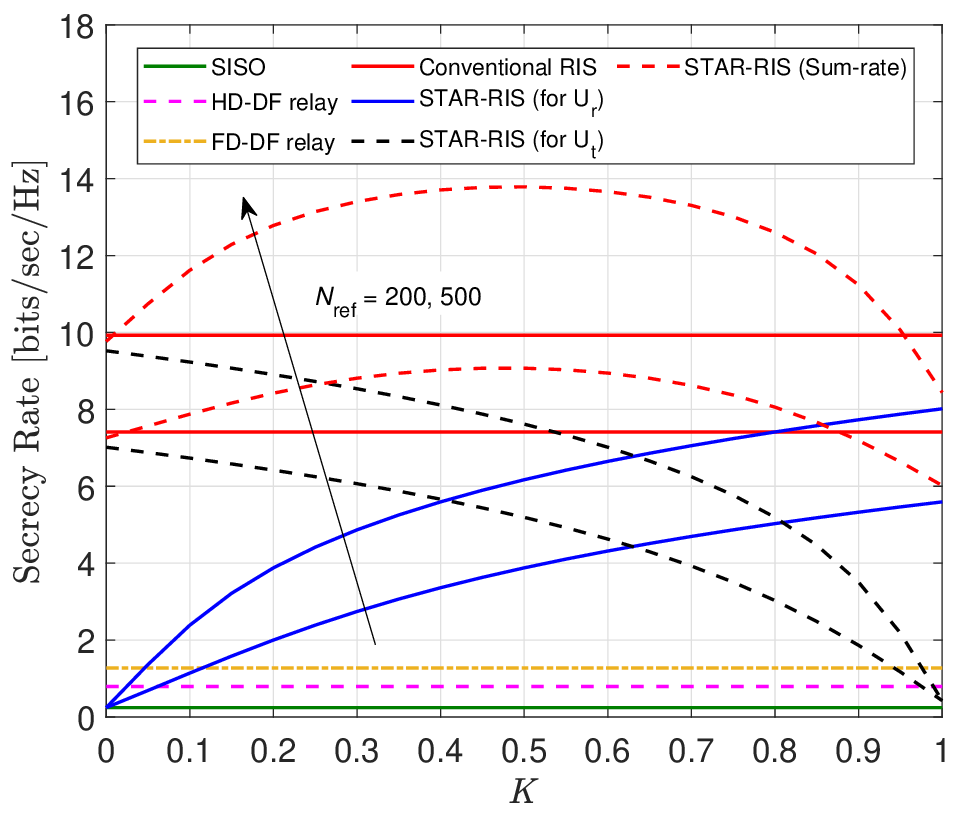}}
\caption{Secrecy rate performance versus $K$, with $d_{\text{s} \text{d}_\text{r}}$ = $d_{\text{s} \text{d}_\text{t}}$ = 100 m, $d_{\text{s} \text{e}_\text{r}}$ = 110 m,
$d_{\text{s} \text{e}_\text{t}}$ = 120 m, $\it{d}_\text{sr}$ = 90 m, $\it{d}_v$ = 10 m, $\it{d}_{v_e}$ = 12 m, $p$ = 10 dBm, $\zeta$ = 0.5 and varying $N_\mathrm{ref}$.}
\label{figure_13}
\end{figure}

Finally, Fig. \ref{figure_13} depicts the secrecy rate performance of the STAR-RIS, SISO, RIS, HD-DF relaying, and FD-DF relaying systems versus $K$ with different values of $N_\mathrm{ref}$. In this figure, a similar trend is seen as in Fig. \ref{figure_7}; that is both STAR-RIS and RIS systems depend strongly on $K$. Unlike the results observed in Fig. \ref{figure_7}, the secrecy rate of HD- and FD-DF relaying systems is low, and relatively few elements are required for both RIS and the two STAR-RIS configurations to outperform HD- and FD-DF relaying systems. In summary, Fig. \ref{figure_13} demonstrates that both STAR-RIS and RIS systems, unlike classical relaying schemes, strongly depend on $K$ and the number of passive elements. Therefore, their ability to achieve higher secrecy rates with a low number of elements highlights their scalability and performance benefits over SISO, HD-, and FD-DF relaying systems.

\textit{To sum up, a set of numerical simulations is presented in Figs. \ref{figure_3}–\ref{figure_13} to validate the analytical findings. These results illustrate both the achievable rate (Figs. \ref{figure_3}–\ref{figure_7}) and secrecy rate (Figs. \ref{figure_8}–\ref{figure_13}) performance of STAR-RIS, SISO, RIS, HD-DF relaying, and FD-DF relaying systems under various network conditions. Specifically, the impact of key system parameters is investigated, including the S-D distance (Figs. \ref{figure_3}, \ref{figure_9}, and \ref{figure_10}), transmit power (Figs. \ref{figure_4} and \ref{figure_8}), reflection-to-transmission power ratio $\zeta$ (Figs. \ref{figure_5} and \ref{figure_11}), the number of reflecting elements $N_\mathrm{ref}$ (Figs. \ref{figure_6} and \ref{figure_12}), and the element-splitting factor $K$ (Figs. \ref{figure_7} and \ref{figure_13}). Through these comparisons, the scalability and security advantages of STAR-RIS-assisted systems over conventional RIS and DF relaying benchmarks are highlighted.}

\section{Conclusion}\label{sec5}

In this paper, the emerging STAR-RIS technology was compared to HD- and FD-DF relaying schemes, SISO and conventional RIS-supported systems in the absence and presence of Es. Simulation results showed that tens of passive elements in RIS and STAR-RIS-supported systems are needed to surpass HD-DF relaying scheme, while hundreds of elements are needed to surpass the FD-DF relaying. In contrast, the secrecy rate analysis indicated that a relatively small number of elements is required for the RIS and STAR-RIS systems to outperform the classical FD-DF relaying schemes. Overall, the findings in this paper confirm that STAR-RIS systems provide substantial secrecy rate advantages over conventional benchmarks, particularly under optimized configurations. This, however, makes them highly promising solutions for secure future wireless networks.

\bibliographystyle{IEEEtran}

\bibliography{Ref1}

\end{document}